\documentclass[twocolumn,superscriptaddress]{revtex4-2}

\usepackage{amsmath}
\usepackage{amssymb}
\usepackage{amsthm}
\usepackage{mathtools}
\usepackage{graphicx}
\usepackage{color}
\usepackage{hyperref}
\usepackage{algorithm}
\usepackage{algpseudocode}
\usepackage{tabularx}
\usepackage{bm}

\newcommand{\be}{\begin{equation}}
\newcommand{\ee}{\end{equation}}
\newcommand{\RR}{\mathbb{R}}

\newcommand{\ZZ}{\mathbb{Z}}
\newcommand{\OO}[1]{\mathcal{O}(#1)}
\newcommand{\paren}[1]{\left( #1 \right)}
\newcommand{\brak}[1]{\left[ #1 \right]}
\newcommand{\cbrak}[1]{\left\{ #1 \right\}}
\newcommand{\abs}[1]{\left| #1 \right|}
\newcommand{\pd}[2]{\frac{\partial#1}{\partial#2}}

\newcommand{\wb}[1]{\overline{#1}}

\newcommand\norm[1]{\left\lVert#1\right\rVert}

\renewcommand{\Re}{\operatorname{Re}}
\renewcommand{\Im}{\operatorname{Im}}

\DeclareMathOperator\Tr{Tr}

\newcommand{\intbz}{\int_{\text{BZ}}}
\newcommand{\bz}{\text{BZ}}
\newcommand{\ibz}{W}
\newcommand{\xtrial}[1]{#1_\text{trial}}
\newcommand{\xtest}[1]{#1_\text{test}}
\newcommand{\vect}[1]{\bm{#1}}

\newtheorem{theorem}{Theorem}
\newtheorem{lem}{Lemma}

\newcommand{\CCQ}{Center for Computational Quantum Physics, Flatiron Institute,
162 5th Avenue, New York, NY 10010, USA}
\newcommand{\CCM}{Center for Computational Mathematics, Flatiron Institute, 162
5th Avenue, New York, NY 10010, USA}

\makeatletter
\newcommand{\multiline}[1]{%
  \begin{tabularx}{\dimexpr\linewidth-\ALG@thistlm}[t]{@{}X@{}}
    #1
  \end{tabularx}
}
\makeatother

\begin{document}

\title{Automatic, high-order, and adaptive algorithms for Brillouin zone integration}

\author{Jason Kaye}
\email{jkaye@flatironinstitute.org}
\affiliation{\CCQ}
\affiliation{\CCM}

\author{Sophie Beck}
\email{sbeck@flatironinstitute.org}
\affiliation{\CCQ}

\author{Alex Barnett}
\affiliation{\CCM}

\author{Lorenzo Van Mu\~noz}
\affiliation{Department of Physics, Massachusetts Institute of Technology, 77
Massachusetts Avenue, Cambridge, MA 02139, USA}

\author{Olivier Parcollet}
\affiliation{\CCQ}
\affiliation{Universit\'e Paris-Saclay, CNRS, CEA, Institut de Physique Th\'eorique, 91191,
   Gif-sur-Yvette, France}

\begin{abstract}
  We present efficient methods for Brillouin zone integration with a non-zero
  but possibly very small broadening factor $\eta$, 
  focusing on cases in which downfolded Hamiltonians can be evaluated
  efficiently using Wannier interpolation.
  We describe robust, high-order accurate algorithms automating convergence to a
  user-specified error tolerance $\varepsilon$, emphasizing an efficient
  computational scaling
  with respect to $\eta$. After analyzing
  the standard equispaced integration method, applicable in the case
  of large broadening, we
  describe a simple iterated adaptive integration algorithm effective in the
  small $\eta$ regime. Its computational cost
  scales as $\OO{\log^3(\eta^{-1})}$ as $\eta \to 0^+$ in three
  dimensions, as opposed to $\OO{\eta^{-3}}$ for equispaced
  integration. We argue that, by contrast, tree-based
  adaptive integration methods scale only as $\OO{\log(\eta^{-1})/\eta^{2}}$ for
  typical Brillouin zone integrals. In
  addition to its favorable scaling, the iterated adaptive algorithm is straightforward to
  implement, particularly for integration on the irreducible Brillouin
  zone, for which it avoids the tetrahedral meshes required for tree-based
  schemes.
  We illustrate the algorithms by calculating the spectral function of
  SrVO$_3$ with broadening on the meV scale.
\end{abstract}

\maketitle

\section{Introduction} \label{sec:intro}

Brillouin zone (BZ) integration is a fundamental operation in
electronic structure calculations for periodic solids, and requires
algorithms capable of accurately computing observables of physical
systems with widely differing properties~\cite{Kratzer/Neugebauer:2019}. In
some cases, BZ integration arises as a post-processing step, for example in the
calculation of the spectral function and optical response functions. In
others, it represents a fundamental step of an iteration procedure,
and must be carried out repeatedly with controlled accuracy. Examples of the
latter include the self-consistent
field calculation of ground state properties within density functional
theory~\cite{Kohn:1999}, and the self-consistency loop in
dynamical mean-field theory (DMFT)~\cite{Georges:1996}. In both cases, BZ integration can represent a
computational bottleneck or a significant source of error, particularly when
high resolution is required to resolve fine features in reciprocal space.

In this article, we focus on many-body Green's function methods, in
which the BZ integral acquires a system- and temperature-dependent
scattering rate which largely determines the difficulty of integration.
Although the ideas discussed here may be
applied to other types of BZ integrals, we use
the single
particle retarded Green's function as a concrete prototypical example. It is
given by 
\begin{equation} \label{eq:gfun}
  G(\omega) = \intbz d^d \vect{k} \Tr \brak{\paren{\omega+\mu-H(\vect{k})-\Sigma(\omega)}^{-1}},
\end{equation}
where BZ denotes the Brillouin zone, $\vect{k}$ is a
reciprocal space vector, $\omega$ is a frequency variable (we
take $\hbar = 1$), $\mu$ is the chemical
potential, $H(\vect{k})$ is a Hermitian Hamiltonian matrix, $\Sigma(\omega)$ is a complex-valued
self-energy matrix, 
the dimension $d\in\{1,2,3\}$, and
scalars are understood to be multiples of the identity matrix. Although in general $\Sigma =
\Sigma(\vect{k},\omega)$, we assume here that the self-energy
is local in space,
$\Sigma =
\Sigma(\omega)$, a situation encountered in a variety of applications, in
particular DMFT. The generalization to the case of
non-local self-energies is problem-dependent and outside the scope of this work. 
The $\vect{k}$-integrated spectral function is
obtained from \eqref{eq:gfun} as
\begin{equation} \label{eq:dos}
  A(\omega) = -\frac{1}{\pi} \Im G(\omega).
\end{equation}
To simplify the discussion we work with a special case of \eqref{eq:gfun}, with the self-energy $\Sigma(\omega) =
-i \eta$ representing a constant scattering rate:
\begin{equation} \label{eq:bzint}
  G(\omega) = \intbz d^d \vect{k} \Tr \brak{\paren{\omega-H(\vect{k})+i\eta}^{-1}}.
\end{equation}
Here, $\mu$ has been absorbed into $H(\vect{k})$. We emphasize, however, that our
algorithms can be straightforwardly applied to the more general case.

\begin{figure*}
  \includegraphics[width=0.25\textwidth,height=0.22\textwidth,keepaspectratio]{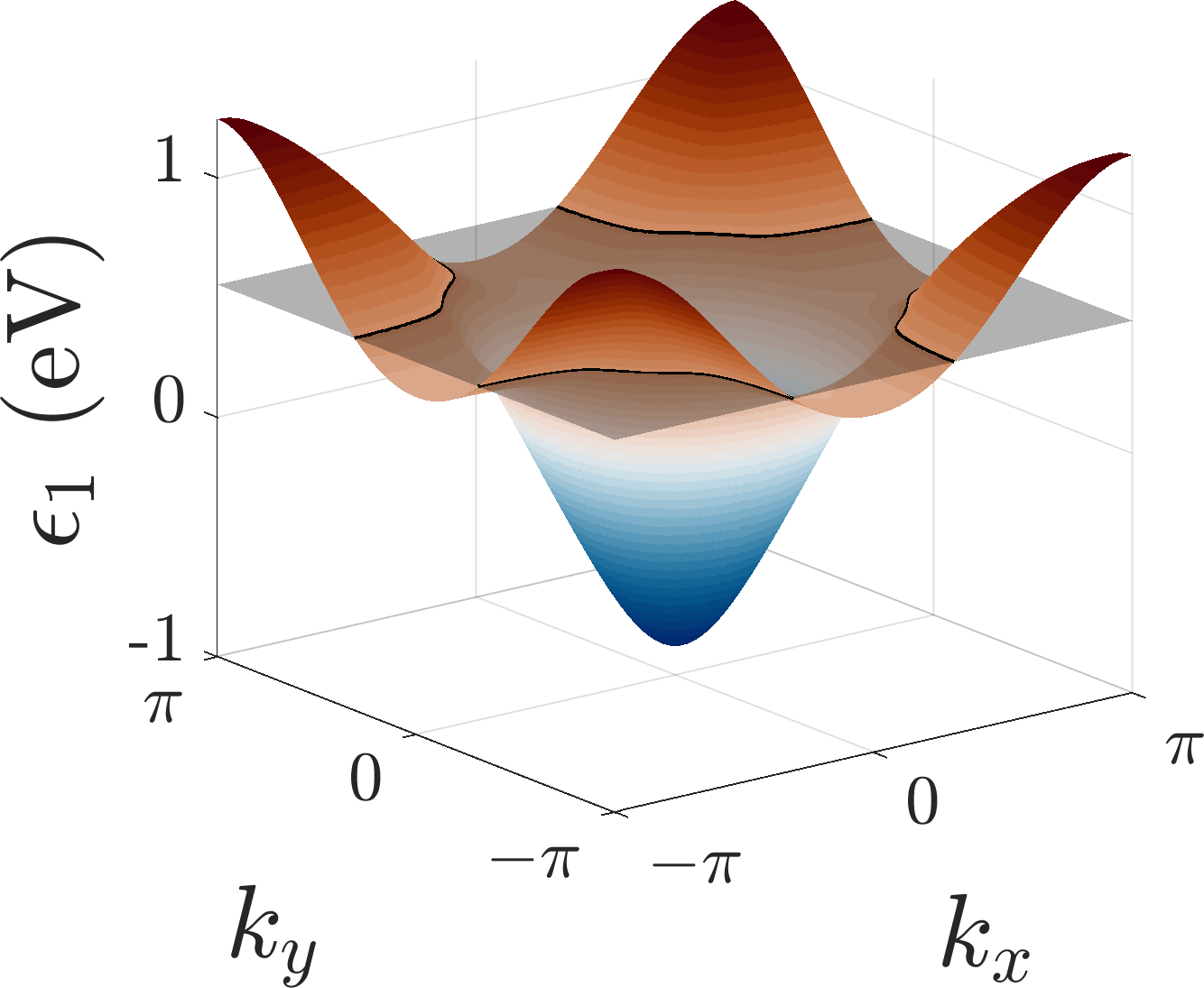}%
  \hfill
  \includegraphics[width=0.25\textwidth,height=0.22\textwidth,keepaspectratio]{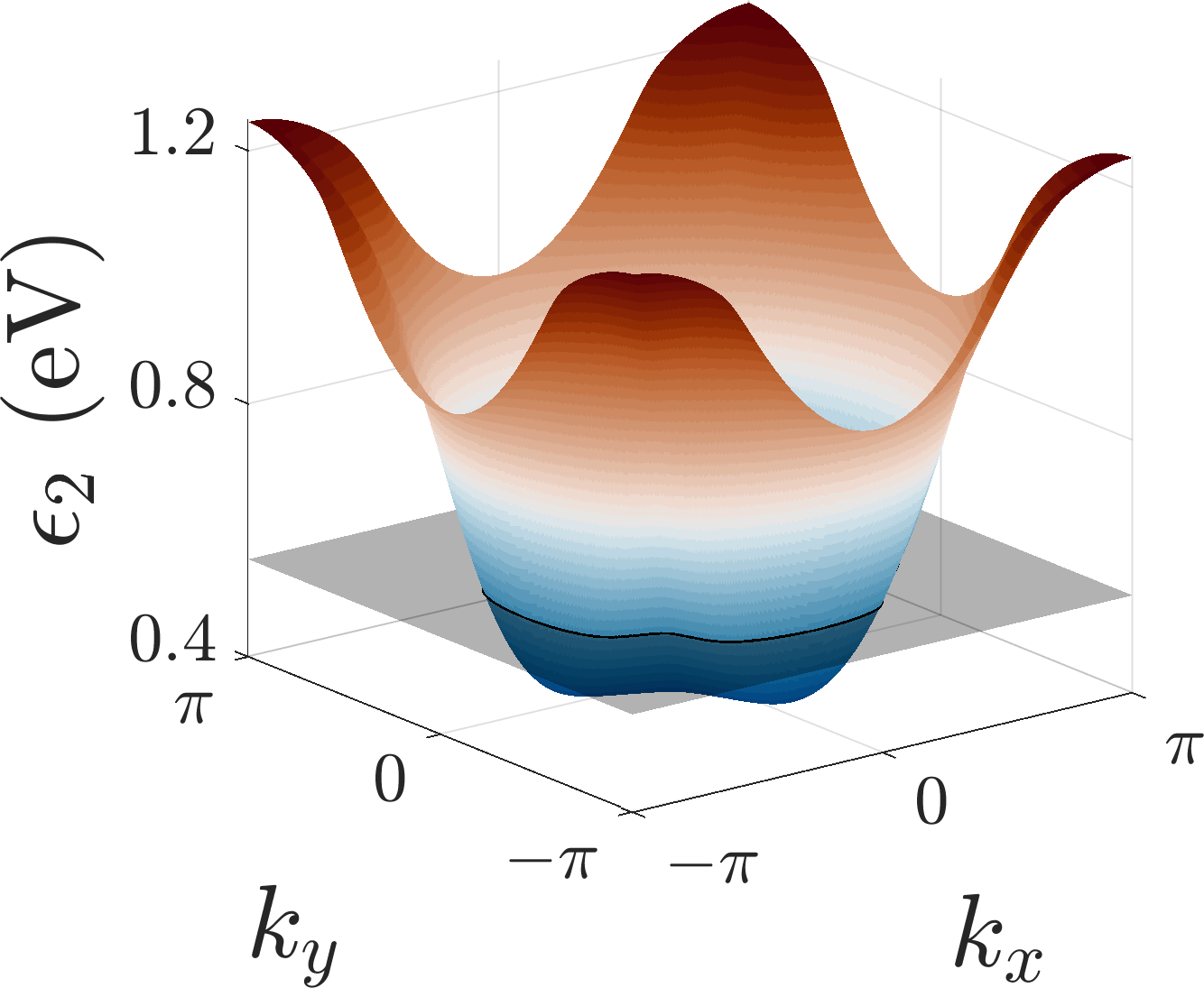}%
  \hfill
  \includegraphics[width=0.25\textwidth,height=0.22\textwidth,keepaspectratio]{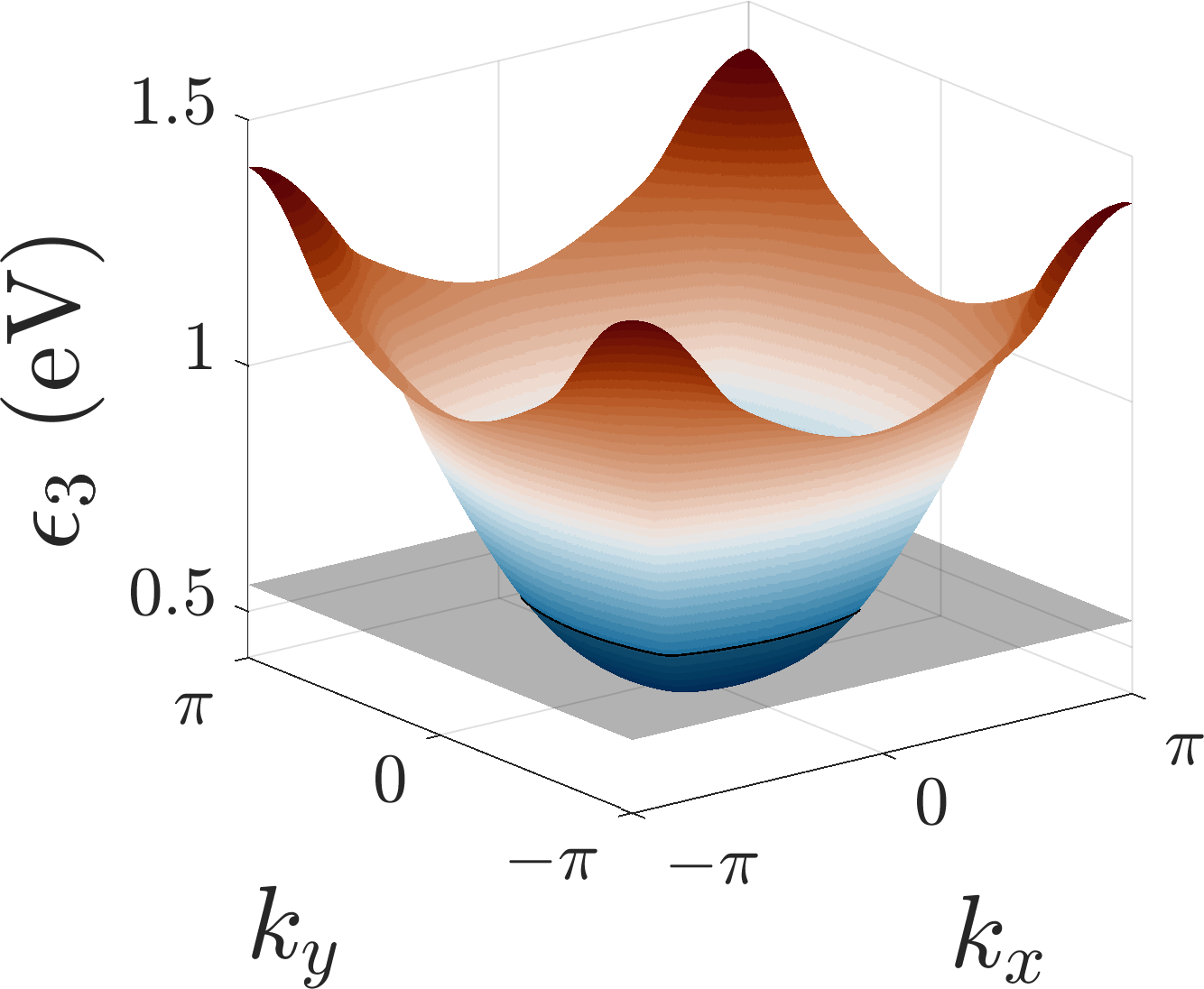}%
  \hfill
  \includegraphics[width=0.25\textwidth,height=0.22\textwidth,keepaspectratio]{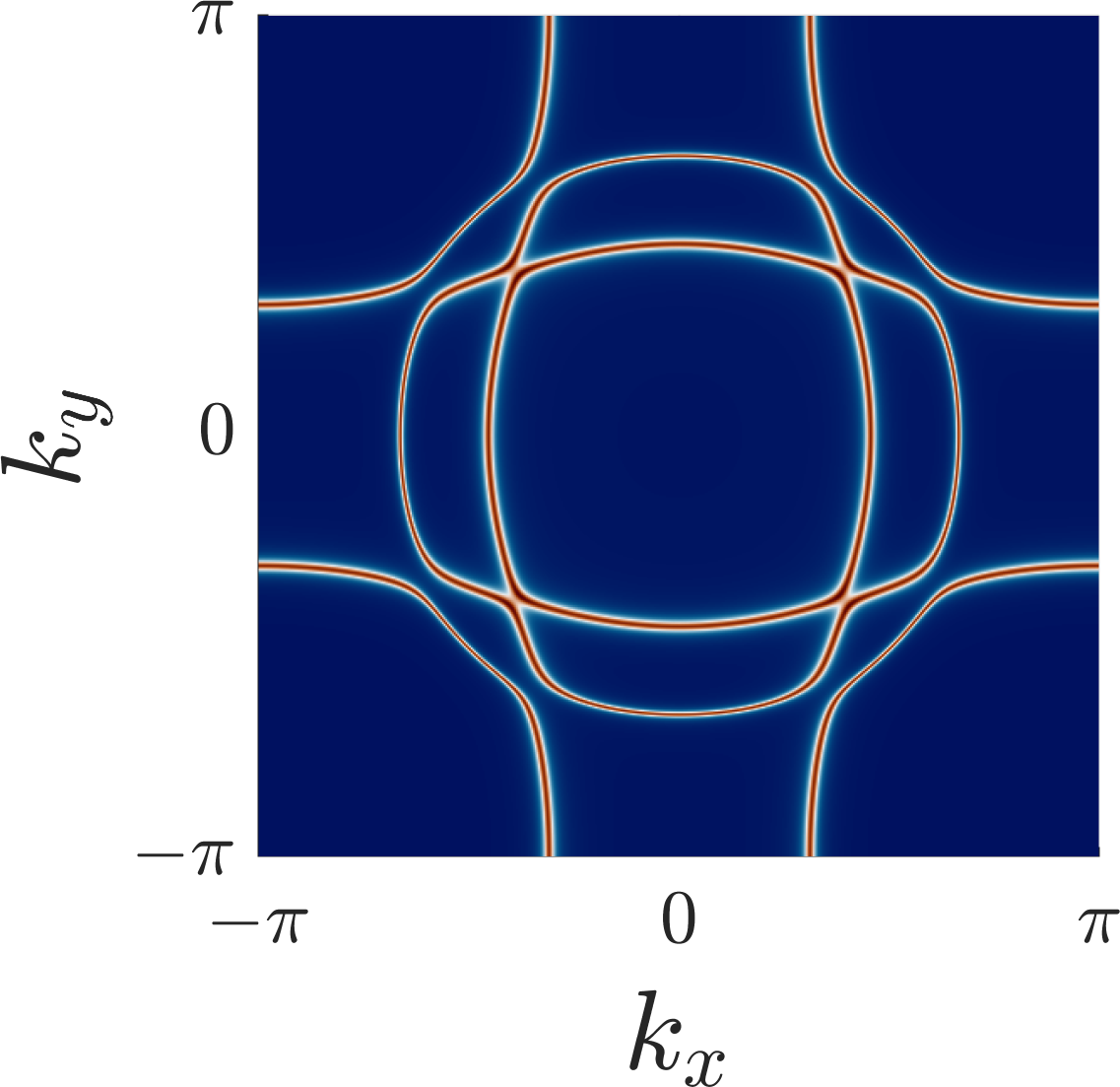}
  \caption{The first three panels show the energy sheets $\epsilon_1$,
    $\epsilon_2$, and $\epsilon_3$ (eigenvalues of $H(\vect{k})$) for the
    SrVO$_3$ example discussed in Section~\ref{sec:results},
    with fixed $k_z = 1.9$. A plane of
  constant $\omega = 0.55$ eV is also shown. The last panel
  shows a color plot of the imaginary part of the integrand in
  \eqref{eq:bzint} with this $\omega$, and $\eta = 0.01$ eV. The integrand
  concentrates in a region of width $\OO{\eta}$ along
  the singular set
  defined by
  $\det\paren{\omega - H(k_x,k_y,1.9)} = \prod_{i=1}^3 \paren{\omega
    - \epsilon_i(k_x,k_y,1.9)} = 0$,
  i.e., the intersections of the plane of constant
  $\omega$ with the three sheets.}
  \label{fig:structure}
\end{figure*}

In much of the literature, quantities of interest are
considered in the limit of zero scattering (i.e., the non-interacting density
of states is obtained from the $\eta = 0^+$ limit of \eqref{eq:bzint}).
In the presence of a self-energy, the scattering rate $\eta$ is
non-zero, but may be small.
For example, a local Fermi liquid has
the asymptotic scaling~\cite{Berthod_et_al:2013}
\begin{equation} \label{eq:fermiliquid}
  \Im\Sigma(\omega, T) \;\sim\; -\brak{\omega^2 + (\pi k_{\text{B}}T)^2}
\end{equation}
as $\omega$, $T \to 0$, with $T$ the temperature, yielding a vanishing but
non-zero scattering rate.
In this case we obtain integrands in \eqref{eq:bzint} with highly localized features at the scale $\OO{\eta}$ near the surface on which
$\det\paren{\omega - H(\vect{k})} = 0$ (e.g.\ the Fermi surface when $\omega =
0$), as illustrated in Fig.~\ref{fig:structure}. In order to compute \eqref{eq:bzint} accurately, these
localized features must be resolved, and adaptive methods become
essential.
In this article we present automatic algorithms to compute BZ
integrals like \eqref{eq:bzint}, with possibly very small $\eta > 0$, to
a user-specified error tolerance $\varepsilon$, with
a cost scaling mildly with respect to
$\eta$ and $\varepsilon$. Mild scaling with respect to $\varepsilon$ is achieved
through the use of high-order accurate methods. The term 
``automatic'' specifies that convergence to error $\varepsilon$ is
carried out by the algorithm itself, not by the user.
Indeed, within self-consistency loops,
it is
crucial that accurate results are returned reliably
without user intervention.

An important distinction in BZ integration algorithms concerns the
representation of $H(\vect{k})$. In many practical applications, the relevant physics can be extracted
from a projection of the ab initio Hamiltonian onto a low-energy
subspace---a process referred to as downfolding---so that $H(\vect{k})$
in \eqref{eq:bzint} becomes a small matrix. $H(\vect{k})$ is then
related to a tight-binding Hamiltonian $H_{\vect{R}}$ by the Fourier series
\begin{equation} \label{eq:fs}
  H(\vect{k}) = \frac{1}{(2\pi)^d} \sum_{\vect{R}} e^{i \vect{k} \cdot \vect{R}} H_{\vect{R}}.
\end{equation}
Here, $\vect{R}$ labels the real space lattice vectors of a tight-binding
model, and mathematically the $H_{\vect{R}}$ are simply the matrix-valued
Fourier coefficients of the periodic function $H(\vect{k})$.
The gauge freedom in the Hamiltonian can be
chosen, using (maximally) localized Wannier functions, to make $H_{\vect{R}}$
decay as rapidly as possible~\cite{Marzari/Vanderbilt:1997,
  Souza/Marzari/Vanderbilt:2001, wannier90}.
$H(\vect{k})$ can then be
approximated
efficiently using a truncation of \eqref{eq:fs}. This procedure is
typically referred to as Wannier interpolation~\cite{Yates_et_al:2007}.
We will
not consider cases in which $H(\vect{k})$ is a large matrix, 
or cases in which Wannier interpolation is not applicable, as other
algorithmic considerations then take precedence.

The literature on BZ integration is extensive, and we mention a few related
approaches here.
A standard method is to sum over a grid of equispaced nodes (such as the Monkhorst--Pack grid~\cite{Monkhorst/Pack:1976}), which we refer to
as the $d$-dimensional
periodic trapezoidal rule (PTR). Despite its spectral accuracy,
to be discussed
shortly,
its obvious drawback is its inefficiency for
small broadening $\eta$. A popular approach for the $\eta = 0^+$ limit of
\eqref{eq:bzint} is the linear tetrahedron method (LTM)
\cite{jepson71,lehmann72,blochl94,kawamura14,ghim22}, which uses a
tetrahedralization of the BZ with piecewise linear approximation of band surfaces
to obtain a semi-analytical result.
There are relatively few works describing the extension of the method to
$\eta > 0$, or more generally to
\eqref{eq:gfun}~\cite{haule10}. The LTM is low-order accurate, and further research is necessary
to clarify its robustness, and its scaling with $\eta$.
We give a brief discussion of the challenges and open questions associated with the LTM in Appendix \ref{app:tm}.
We
also mention smearing methods \cite{ho82,methfessel89}, used for the $\eta = 0^+$ case, which avoid
computing distributional integrals appearing in the $\eta=0^+$ limit
by effectively adding a
small broadening $\eta$ and then applying a
uniform integration scheme. Adaptive
smearing methods involve specific prescriptions for possibly
$\vect{k}$-dependent broadening
parameters \cite{bjorkman11,Yates_et_al:2007}; note that
this is distinct from adaptive integration in the sense used
in this article. Several other integration schemes have been proposed recently
\cite{wang21,chen22,duchemin22},
but do not primarily address the class of problems considered here.

A few low-order accurate adaptive integration algorithms have been proposed. Refs.~\cite{henk01} and \cite{assmann16} describe tree-based adaptive
integration methods, which are discussed below. Refs.~\cite{henk01} and
\cite{bruno97} describe iterated adaptive integration algorithms which are
close to that proposed here, but they have not gained
widespread adoption. Our high-order accurate version of
this approach, which includes critical performance optimizations,
significantly improves its competitiveness.

\begin{figure*}
  \includegraphics[width=0.33\textwidth,height=0.33\textwidth,keepaspectratio]{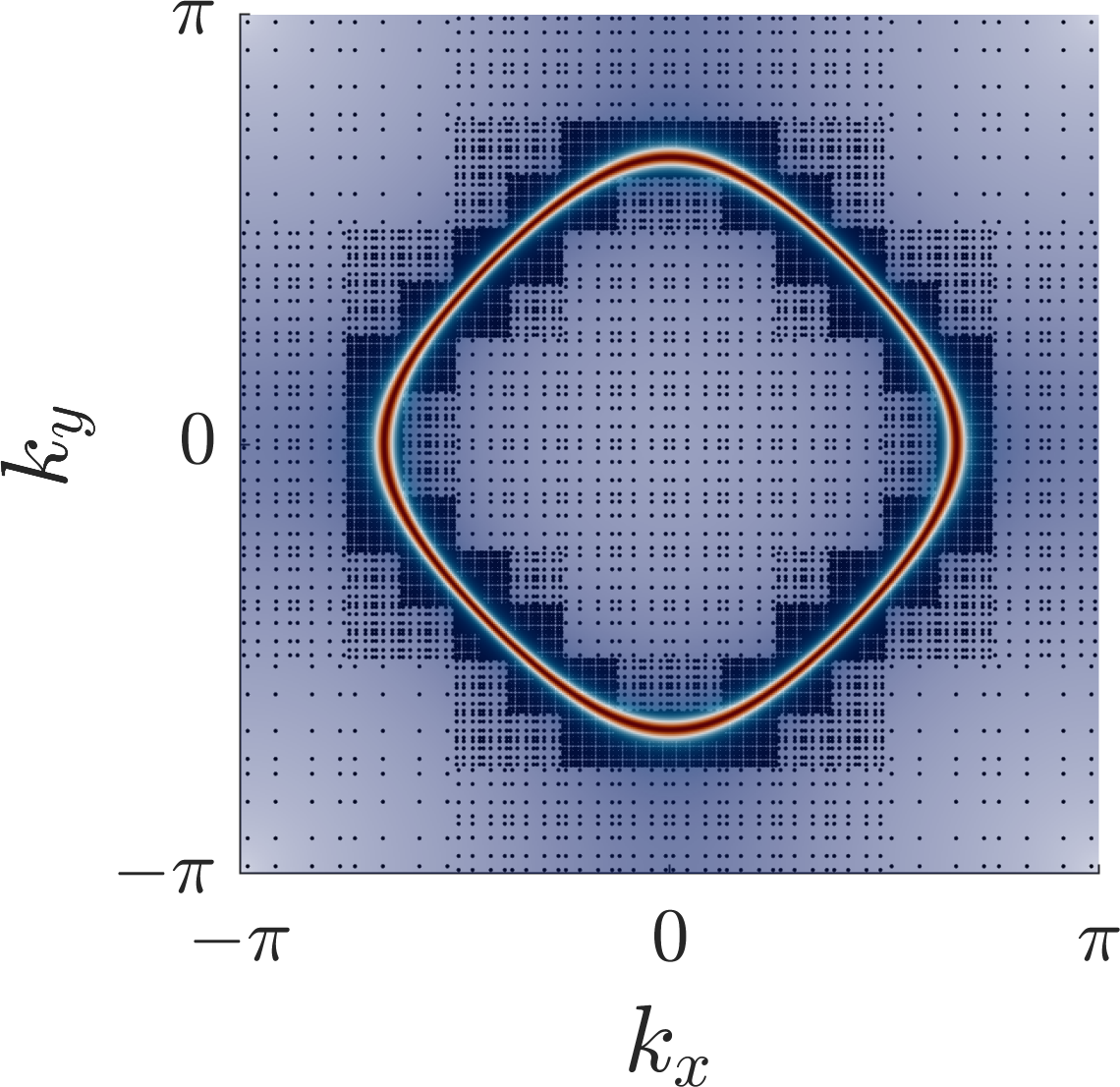}%
  \hfill
  \includegraphics[width=0.33\textwidth,height=0.33\textwidth,keepaspectratio]{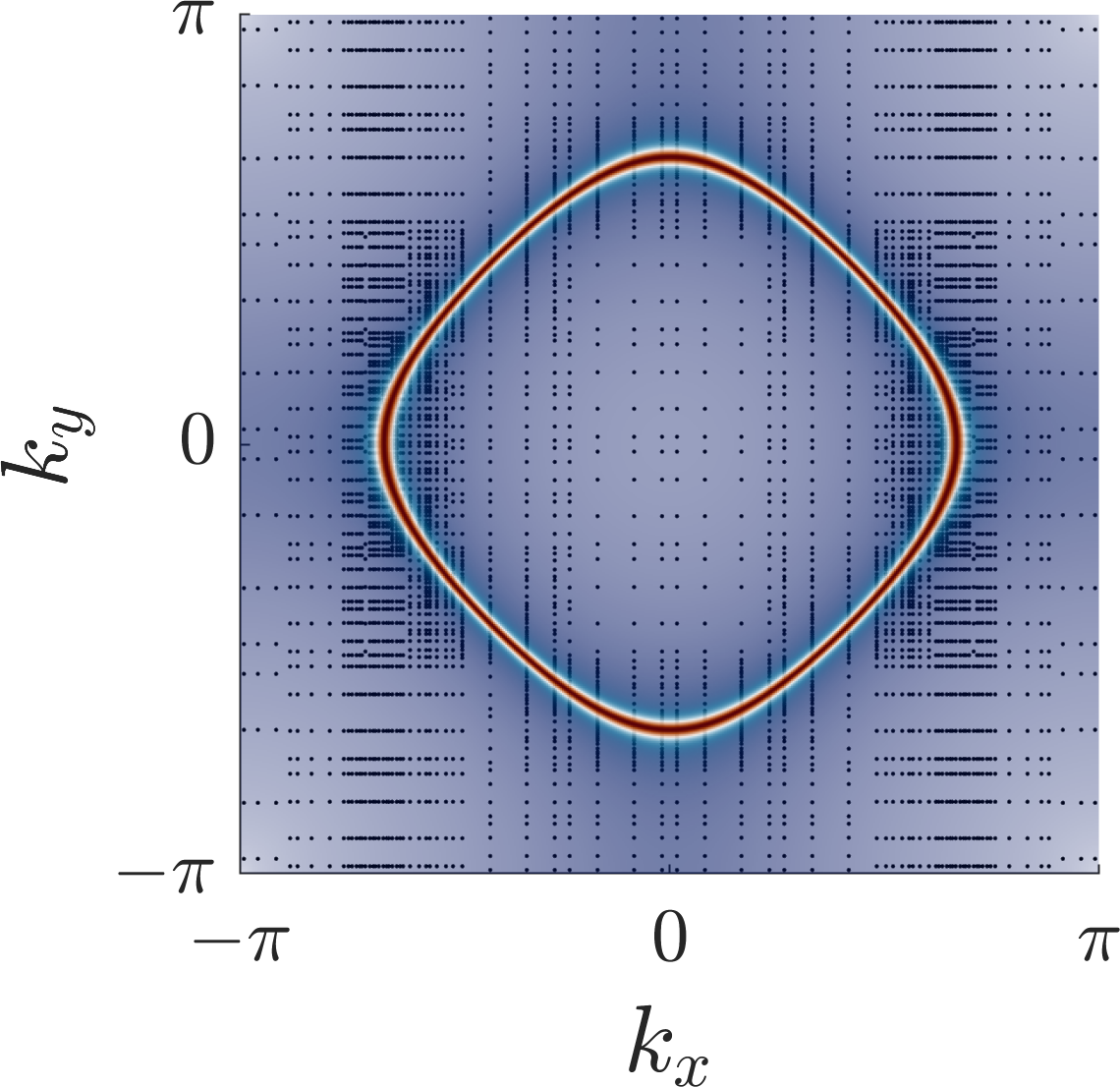}%
  \hfill
  \includegraphics[width=0.33\textwidth,height=0.33\textwidth,keepaspectratio]{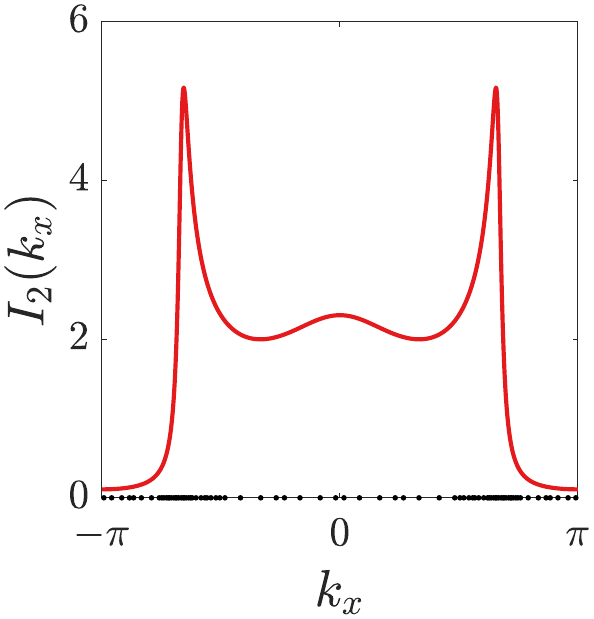}
  \caption{
    Adaptive integration grids for the spectral function \eqref{eq:dos} with $G$ given by \eqref{eq:bzint} for
  $d = 2$, $H(k_x,k_y) = \cos \paren{k_x} + \cos \paren{k_y}$, $\omega = 0.5$ eV,
  $\eta = 0.05$ eV, and specified error tolerance $\varepsilon = 10^{-3}$
  (in units of the spectral function). The TAI grid
  (left) refines along the full nearly-singular curve, yielding $\OO{\log(\eta^{-1})/\eta}$
  points (or $\OO{\log(\eta^{-1})/\eta^2}$ in 3D). The integrand is shown as a
  color plot underneath the grid. By contrast, the IAI grid (middle) only refines
  into two nearly-singular points in $k_x$ (shown in the right panel), and then for each $k_x$,
  into at most two nearly-singular points in $k_y$, yielding
  $\OO{\log^2(\eta^{-1})}$ points in total (or $\OO{\log^3(\eta^{-1})}$ in
  3D). The right panel shows the outer integrand $I_2(k_x) =
  -\frac{1}{\pi} \Im \int_{-\pi}^{\pi} dk_y \, \paren{\omega-H(k_x,k_y) + i \eta}^{-1}$, with its adaptive integration grid in $k_x$ shown as black
  dots.}
  \label{fig:grids}
\end{figure*}

The main contributions of this work are as follows:
\begin{itemize}
  \item We present a simple, high-order accurate \textit{iterated} adaptive
    integration (IAI) method, relying only on the use of a one-dimensional
    (1D) adaptive integration algorithm. It has $\OO{\log^d(1/\eta)}$ computational
complexity as $\eta \to 0^+$. This is to be contrasted
with the $\OO{\eta^{-d}}$ scaling of uniform
    integration methods like the PTR, which require an $\OO{\eta}$ grid spacing to resolve
    features of width $\eta$.
  \item We observe that \textit{tree-based}
adaptive integration (TAI), another common approach \cite{henk01,assmann16},
scales
    only as $\OO{\log(\eta^{-1})/\eta^{d-1}}$, and is therefore asymptotically slower than IAI
for small $\eta$. An example of the difference between the grids produced by TAI and IAI is
    given in Fig.~\ref{fig:grids} for the simple 2D example
    $H(k_x,k_y) = \cos(k_x) + \cos(k_y)$.
    Whereas TAI requires building
    adaptive tetrahedral meshes for calculations on the irreducible Brillouin
    zone (IBZ), IAI can be
    used directly on the IBZ with minimal modification.
  \item We include an in-depth discussion of the PTR, emphasizing its
    high-order accuracy. Despite its $\OO{\eta^{-d}}$ scaling, the PTR has significant
    advantages in the large $\eta$ regime. We show how
    the PTR can be
    used in an automated manner, and conclude that a
hybrid approach, using the PTR for large $\eta$ and IAI for small
$\eta$, provides an efficient and robust solution covering the full range of cases.
\end{itemize}
In addition, we present a high-order adaptive frequency interpolation method
    to represent the spectral function. This algorithm
    automatically generates an efficient grid to resolve localized features
    like Van Hove singularities, thereby minimizing the number of BZ
    integrals which need to be computed to obtain $A(\omega)$.
We demonstrate our method by calculating the
    spectral function of SrVO$_3$ with constant scattering rates
    $\eta$ as small as $1$ meV to several digits of accuracy.

This article is organized as follows. We begin in Section~\ref{sec:ptr}
with a discussion of the PTR and its advantages when $\eta$ is not too
small. We review high-order 1D adaptive integration in
Section~\ref{sec:adapintgr}, and compare the IAI and TAI approaches to
higher-dimensional integration in Section~\ref{sec:adapintnd},
concluding that IAI is superior for BZ integration.
We describe our
automatic adaptive frequency sampling method in Section~\ref{sec:adapinterp}, and demonstrate the performance of the full scheme
for a calculation of the spectral function of SrVO$_3$ in 
Section~\ref{sec:results}. Application areas and future
directions of research are discussed in Section~\ref{sec:conclusion}. We
also include several appendices: Appendix
\ref{app:ptrtheory} presents analysis of the PTR for BZ
integrals, Appendix \ref{app:ibz} gives
details on the use of the PTR for integration in the IBZ, 
Appendix \ref{app:ptrheval} discusses the efficient
implementation of the PTR, Appendix \ref{app:autoptralt} gives an alternative
algorithm for automatic refinement of the PTR to that presented in the main
text, and Appendix \ref{app:tm} discusses the LTM.

\subsection*{Terminology and notation}

We briefly fix some standard notation and terminology, and refer to Ref.
\cite[Chap. 3]{kaxiras03} for useful background.
For concreteness, we restrict here to the 3D case.

The reciprocal lattice vectors $\vect{b}_1,\vect{b}_2,\vect{b}_3 \in \RR^3$
are related to the primitive lattice vectors
$\vect{a}_1,\vect{a}_2,\vect{a}_3 \in
\RR^3$ by $\vect{a}_i \cdot \vect{b}_j = 2 \pi
\delta_{ij}$,
or equivalently, writing $A = \begin{bmatrix} \vect{a}_1 &
\vect{a}_2 & \vect{a}_3 \end{bmatrix}$ and $B = \begin{bmatrix}
  \vect{b}_1 &
\vect{b}_2 & \vect{b}_3 \end{bmatrix}$,
by $A^T B = 2\pi I$, with $I$ the identity
matrix. All $\vect{k}$-dependent quantities, like $H(\vect{k})$, are periodic
with respect to translation by the reciprocal lattice vectors, i.e.,
$H(\vect{k}+B\vect{n}) = H(\vect{k})$ with $\vect{n} \in \ZZ^3$.
The term ``Brillouin zone'' often refers to the first Brillouin zone, the cell in
the Voronoi decomposition of the reciprocal lattice containing the origin. However, this domain is in
general a complicated polytope, and therefore less convenient for
integration than the parallelpiped spanned by the reciprocal lattice
vectors.
Using the freedom afforded to us by the periodicity, we
take this as our primitive unit cell, and refer to it as the Brillouin zone:
\[\bz \equiv
\cbrak{\vect{k} = B\vect{\kappa}/2\pi \, | \, \vect{\kappa} \in
[-\pi,\pi]^3}.\]
By changing to the $\vect{\kappa}$ variables, we see that up to
multiplication by a Jacobian factor $\abs{\det B}/(2\pi)^3$ we can assume
without loss of generality that $\bz = [-\pi,\pi]^3$, or equivalently,
that $A = I$, $B = 2\pi I$, and $\vect{k} = \vect{\kappa}$. For
simplicity of exposition, we make this assumption in the remainder of
the article, except in discussions of IBZ integration, for
which the specific point symmetries determine the integration domain.

The Fourier series representation \eqref{eq:fs} will be used to evaluate
$H(\vect{k})$ efficiently. Under the assumptions made above, 
$\vect{R}$ is an integer multi-index, $\vect{R} = (R_1,R_2,R_3) \in \ZZ^3$. Truncating the rapidly
converging Fourier
series to include $M$ modes per dimension, we obtain
\begin{equation} \label{eq:fstrunc}
  H(\vect{k}) \approx \frac{1}{(2\pi)^3} \sum_{R_1,R_2,R_3 = -M/2}^{M/2-1} e^{i
  \vect{k} \cdot \vect{R}} H_{\vect{R}},
\end{equation}
where for simplicity we have assumed $M$ is even.

The spectral function has units of (eV$\mbox{\AA}^3$)$^{-1}$, but we suppress
them throughout the text.

\section{Periodic trapezoidal rule} \label{sec:ptr}

As a starting point, we consider \eqref{eq:bzint} with $\eta$
not too small.
Although this case is less challenging than that of small $\eta$, the
two scenarios are often encountered side by side in practice, for example in the
presence of a self-energy $\Sigma(\omega)$ with $\abs{\Im \Sigma(\omega)}$ large for 
some $\omega$ and small for others. Taking the $\bz$ as specified
above, \eqref{eq:bzint} is the integral of a smooth, triply-periodic
function over a cube.

The standard tool in this case is the PTR. For
a $2\pi$-periodic function $f(x)$ of one variable, this is simply the
quadrature rule $\int_0^{2\pi} f(x) \, dx \approx \frac{2\pi}{N}
\sum_{n=0}^{N-1} f(2\pi n/N)$, where we have shifted to the interval
$[0,2\pi]$ for notational convenience. We emphasize that although the ordinary trapezoidal rule
for non-periodic functions is only second-order accurate, the PTR is
\textit{spectrally} accurate. In particular, the following theorem
describes the error for analytic functions~\cite[Thm 3.2]{trefethen14}
(see also \cite{Davis59}).
\begin{theorem} \label{thm:ptr}
  If $f(x)$ is $2\pi$-periodic and analytic in the the strip $\abs{\Im
    x} < a$, with $\abs{f(x)} \leq C$, then for any $N \geq 1$,
  \[\abs{\int_0^{2\pi} f(x) \, dx - \frac{2\pi}{N} \sum_{n=0}^{N-1} f(2\pi
  n/N)} \leq \frac{4 \pi C}{e^{aN}-1}.\]
\end{theorem}

\begin{figure}
  \centering
  \includegraphics[width=0.98\linewidth]{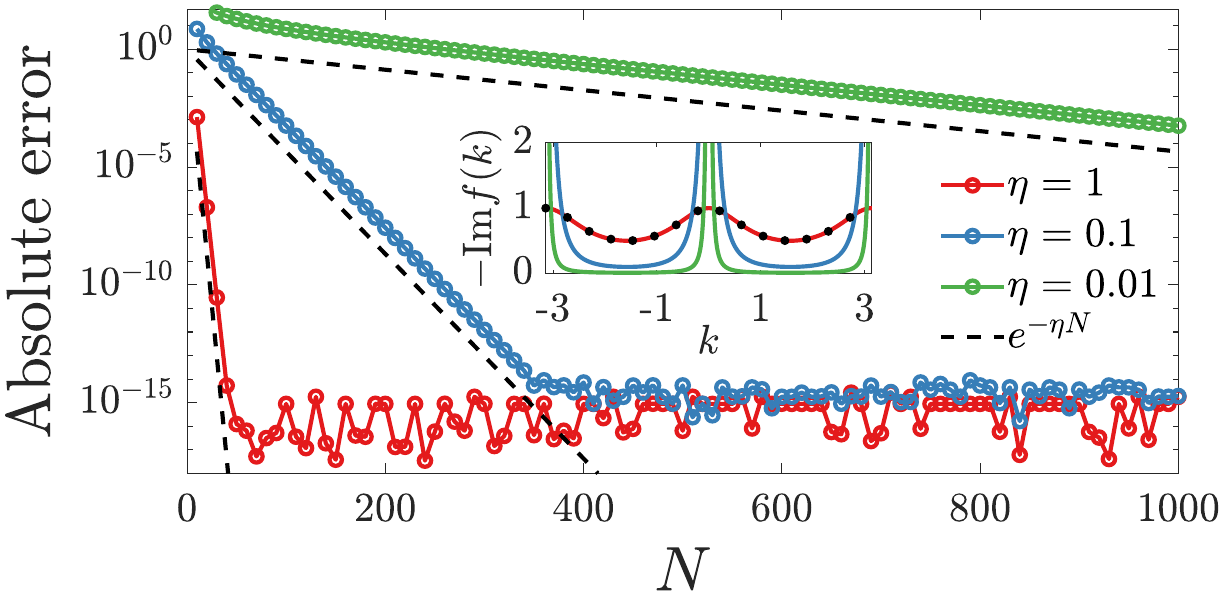}
  \caption{Error of the $N$-point PTR applied
  to $\int_{-\pi}^{\pi} \frac{dk}{\sin(k)+i \eta}$ with three
  values of $\eta$, along with a simple estimate of the convergence
  rate derived from Theorem~\ref{thm:ptr}. The negative imaginary part of the
  integrand $f(k)$ is shown in the inset. For $\eta = 1$,
  a relatively coarse discretization of $N = 15$ nodes (black dots) yields five-digit
  accuracy, and $N = 40$ yields full double-precision accuracy.
  }
  \label{fig:ptrex}
\end{figure}

Thus the exponential
rate of convergence is given by the distance from the
real axis to the closest singularity of $f$ in the complex plane. As an
example, in Fig.~\ref{fig:ptrex}, we plot the error of the PTR against
$N$ for the integral $\int_0^{2\pi} \frac{dk}{\sin(k)+i
\eta}$, a 1D case of \eqref{eq:bzint}, for several
choices of $\eta$. For large $\eta$, we observe exceptionally rapid
convergence, yielding high accuracy with a coarse discretization, but
for small $\eta$ localized features emerge and the convergence is slow.
Indeed, Theorem \ref{thm:ptr} predicts error scaling
close to
$\OO{e^{-\eta N}}$, since a first-order Taylor expansion implies
poles near $k = -i \eta, \pi+i\eta$.
Appendix~\ref{app:ptrtheory} makes this intuition rigorous
for certain integrals of the form \eqref{eq:bzint}.

For $d>1$, we can simply apply the PTR
dimension by dimension; for example, in 2D,
\begin{equation} \label{eq:ptr2d}
  \begin{multlined}
  \int_{[0,2\pi]^2} f(x,y) \, dx \, dy \\ \approx
    \frac{(2\pi)^2}{N_1 N_2} \sum_{n_1,n_2=0}^{N_1-1,N_2-1} f\paren{\frac{2\pi
n_1}{N_1},\frac{2\pi n_2}{N_2}}.
  \end{multlined}
\end{equation}
The analogue of Theorem \ref{thm:ptr} for $d > 1$ is discussed in Appendix
\ref{app:ptrtheory}. We again obtain exponential convergence, with rate
determined by the height of the strip of analyticity in each variable,
minimized over all other variables.

\subsection{Periodic trapezoidal rule for Brillouin zone integrals}
\label{subsec:ptrbz}

The PTR is therefore highly efficient for integrals of the form
\eqref{eq:bzint} when $\eta$ is sufficiently large. First, we can obtain the coefficients $H_{\vect{R}}$ in
\eqref{eq:fstrunc} by Fourier interpolation from samples $H(q_m)$ on a coarse, uniform grid of
$M^d$ points $q_m$. The values $H(q_m)$ are typically obtained from an ab
initio calculation of the band eigenvalues $\epsilon_n(k)$, transformed into
the optimized Wannier basis.
Then, using \eqref{eq:fstrunc},
$H(\vect{k})$ can be interpolated to a uniform integration
grid of $N^d$ points $k_n$ and stored.
Typically $H(\vect{k})$ is much smoother than the integrand in
\eqref{eq:bzint}, so $N \gg M$, i.e., the grid
required to accurately
represent $H(\vect{k})$ can be made much coarser than the BZ integration grid.
Finally, for each $\omega$, the
PTR can be applied to \eqref{eq:bzint}, requiring the
calculation of a matrix inverse at each of $N^d$ grid points.

We emphasize that $H(\vect{k})$ need only be computed once on the integration
grid, stored, and reused for each $\omega$. Since the
cost of evaluating $H(\vect{k})$ using \eqref{eq:fstrunc} may be comparable to or
greater than that of computing a small matrix inverse, this represents an 
advantage of the PTR over adaptive schemes. Indeed, any adaptive method
must evaluate $H(\vect{k})$
on the fly for each new $\omega$, since the structure of an adaptive
grid depends on the
level set
$\det\paren{\omega - H(\vect{k})} = 0$
(see Fig.~\ref{fig:structure}). The efficient evaluation of
\eqref{eq:fstrunc} for the PTR is discussed in Appendix \ref{app:ptrheval}. 

Although the IBZ is in general a complicated polytope (and in particular
not a periodic domain),
it can be shown that the PTR with $N_1 = N_2 = \cdots$ can be exactly
\textit{symmetrized} to a sum only over grid points within the IBZ, while maintaining
spectral accuracy.
Appendix \ref{app:ibz} discusses this point in detail.

\subsection{Automatic integration with the periodic trapezoidal rule}
\label{sec:autoptr}

The PTR can be used in a black-box algorithm to compute $G(\omega)$
in \eqref{eq:bzint} to a user-specified error tolerance $\varepsilon$.

%\begin{algorithm}[H]
%  \caption{Compute $G(\omega)$ in \eqref{eq:bzint} by automatic
%  integration with PTR on full BZ}
%  \label{alg:ptr}
%  \textbf{Inputs:} $H_{\vect{R}}$ (defined by \eqref{eq:fstrunc}), $\eta$,
%  a collection of $n_\omega$ points $\omega_j$, $\varepsilon > 0$,
%  positive integers $\xtrial{N}$ and $\Delta N$ \\
%  \textbf{Output:} $G(\omega_j)$, $j = 1,\ldots,n_\omega$, correct to $\varepsilon$ accuracy
%  \begin{algorithmic}
%    \State Set $\xtest{N} \gets \xtrial{N} + \Delta N$
%  \State Evaluate $H(\vect{k})$ at $\xtrial{N}^d$, $\xtest{N}^d$
%  equispaced nodes, using \eqref{eq:fstrunc}, to obtain arrays $\xtrial{H}$,
%  $\xtest{H}$
%  \For{$j=1,\ldots,n_\omega$}
%    \State \multiline{Compute $G(\omega_j)$ by PTR with
%      $\xtrial{N}^d$, $\xtest{N}^d$ points, \\ using stored arrays
%      $\xtrial{H}$, $\xtest{H}$, to obtain $\xtrial{G}$, $\xtest{G}$}
%    \State $e \gets \abs{\xtrial{G}-\xtest{G}}$
%    \While{$e > \varepsilon$}
%    \State $\xtrial{\{N,H,G\}} \gets \xtest{\{N,H,G\}}$
%      \State $\xtest{N} \gets \xtest{N} + \Delta N$
%      \State Evaluate $H(\vect{k})$ on $\xtest{N}^d$ points to obtain $\xtest{H}$
%      \State \multiline{Compute $G(\omega_j)$ by PTR with $\xtest{N}^d$
%      points to \\ obtain $\xtest{G}$}
%      \State $e \gets \abs{\xtrial{G}-\xtest{G}}$
%    \EndWhile
%    \State Store $G(\omega_j) = \xtest{G}$
%  \EndFor
%  \end{algorithmic}
%\end{algorithm}

\begin{algorithm}[H]
  \caption{Automatic PTR to compute $G(\omega)$}
  \label{alg:ptr}
  \textbf{Inputs:} $\eta > 0$, $\omega$, $\varepsilon > 0$, positive integers
  $\xtrial{N}$ and $\Delta N$, arrays $\xtrial{H}$ and $\xtest{H}$ of values of $H(\vect{k})$ on PTR grids of $N = \xtrial{N}$ and $N = \xtest{N} \equiv \xtrial{N} + \Delta N$ points per
  dimension, $H_{\vect{R}}$ \\
  \textbf{Output:} $G(\omega)$ correct to $\varepsilon$ accuracy, updated
  integers $\xtrial{N}$, $\xtest{N}$ and corresponding arrays $\xtrial{H}$,
  $\xtest{H}$
  \begin{algorithmic}
  \State \multiline{Using $\xtrial{H}$, $\xtest{H}$,
    approximate $G(\omega)$ by $\xtrial{N}^d$, $\xtest{N}^d$-\\point PTR to obtain $\xtrial{G}$, $\xtest{G}$}
  \State $e \gets \abs{\xtrial{G}-\xtest{G}}$
  \While{$e > \varepsilon$}
    %\State $\xtest{N} \gets \xtest{N} + \Delta N$
    %\State $\xtrial{G} \gets \xtest{G}$
    \State $\xtrial{\{N,H,G\}} \gets \xtest{\{N,H,G\}}$
    \State $\xtest{N} \gets \xtrial{N} + \Delta N$
    \State Evaluate $H(\vect{k})$ on $\xtest{N}^d$ points to obtain $\xtest{H}$
    \State \multiline{Compute $G(\omega)$ by $\xtest{N}^d$-point PTR 
    to obtain $\xtest{G}$}
    \State $e \gets \abs{\xtrial{G}-\xtest{G}}$
  \EndWhile
  \State $G(\omega) \gets \xtest{G}$
  \end{algorithmic}
\end{algorithm}

%This algorithm estimates the error using a self-consistency check
%between the results calculated using the PTR with $\xtrial{N}$ and
%$\xtest{N} = \xtrial{N} + \Delta N$ points, respectively. As the algorithm steps
%through the $\omega_j$, the number of grid points is increased if
%necessary to ensure self-consistency to within $\varepsilon$. Since
%evaluating $H(\vect{k})$ is typically a computational bottleneck, the
%calculation for each new value of $\omega$ reuses the arrays
%$\xtrial{H}$ and $\xtest{H}$ corresponding to the largest
%values of $\xtrial{N}$ and $\xtest{N}$ used so far, even
%if a less dense grid is required for good accuracy.
%A refinement of the algorithm might maintain values of $H(\vect{k})$ on
%a collection of grids with different numbers of points, and for each new
%$\omega_j$ attempt to use the appropriate grid first. However, this is unlikely
%to provide a significant improvement as long as the cost of evaluating
%$H(\vect{k})$ is dominant. Most importantly, Algorithm \ref{alg:ptr} ensures that each value
%$G(\omega_j)$ is computed using a grid sufficiently dense to yield error
%below $\varepsilon$.

This algorithm estimates the error using a self-consistency check between the
results calculated using the PTR with $\xtrial{N}$ and $\xtest{N} = \xtrial{N} +
\Delta N$ points, respectively. Although it is stated for a single evaluation
point $\omega$, in practice, 
given a collection of frequency points $\omega$,
we repeatedly apply Algorithm \ref{alg:ptr}, using the output values
$\xtrial{\{N,H\}}$ and $\xtest{\{N,H\}}$ from the previous frequency point as
inputs for the next one. In this way, as $G(\omega)$
is computed for each new $\omega$, the number of grid points is increased if
necessary to ensure self-consistency to within $\varepsilon$. Since evaluating
$H(\vect{k})$ is typically a computational bottleneck, reusing the arrays
$\xtrial{H}$ and $\xtest{H}$ corresponding to the largest values of $\xtrial{N}$
and $\xtest{N}$ used so far is typically efficient, even if a less dense grid is
required for good accuracy for a specific value of $\omega$.  A refinement of
the algorithm might maintain values of $H(\vect{k})$ on a collection of grids
with different numbers of points, and for each new $\omega$ attempt to use the
appropriate grid first. However, this is unlikely to provide a significant
improvement as long as the cost of evaluating $H(\vect{k})$ is dominant. Most
importantly, Algorithm \ref{alg:ptr} ensures that each value $G(\omega)$ is
computed using a grid sufficiently dense to yield error below $\varepsilon$.

We remark that if $\Delta N$ is chosen too small, it is possible for the
self-consistency check to provide an error estimate smaller than the true error,
yielding a result with error greater than $\varepsilon$. The selection of
$\Delta N$ depending on $\eta$ and the properties of $H$ is discussed in
Appendix \ref{app:ptrtheory}. We also note that to avoid taking many refinement
steps in each call of Algorithm \ref{alg:ptr}, the initial value of $\xtrial{N}$ should be
chosen proportional to $\eta^{-1}$. For the examples in Section
\ref{sec:results}, we take $\xtrial{N} = 6/\eta$.

We note that the procedure we have described allows the evaluation frequencies
$\omega$ to be determined on-the-fly, between subsequent calls to Algorithm
\ref{alg:ptr}. This is required for the adaptive frequency interpolation
algorithm described in \ref{sec:adapinterp}, in which the evaluation frequencies
are selected to discover and resolve localized features in the spectral
function.  In Appendix \ref{app:autoptralt}, we suggest an alternative algorithm
which may be advantageous in the alternative case that $G(\omega)$ is to be
evaluated at a given, fixed set of frequencies. 

\section{Adaptive integration preliminaries} \label{sec:adapintgr}

For sufficiently small $\eta$, resolving \eqref{eq:bzint} using the PTR is 
infeasible. We saw a 1D example of this in Fig.~\ref{fig:ptrex}, and refer again
to Fig.~\ref{fig:structure} for an example
of a highly localized integrand in the 2D case.
Adaptive integration 
enables the efficient and automatic evaluation of singular or
nearly-singular integrals with high-order accuracy.
Before discussing higher-dimensional adaptive integration methods in Section~\ref{sec:adapintnd}, we
first review Gauss quadrature for smooth non-periodic integrals in Section~\ref{sec:gq}, and its use  
in 1D high-order adaptive integration in Section~\ref{sec:adapint1d}.

\subsection{Gauss quadrature} \label{sec:gq}

The Gauss quadrature rule is given by 
\[\int_{a}^b f(x) \, dx \approx \sum_{j=1}^p f(x_j) w_j,\]
where $x_j$ and $w_j$ are the nodes and weights of the rule,
respectively. The rapid convergence properties of Gauss quadrature make it a standard
tool for the integration of smooth, non-periodic functions. In particular, it is the unique
$p$-node quadrature rule which integrates polynomials up to degree $2p-1$
exactly. The following theorem describes the error for
analytic functions \cite[Thm 4.5]{trefethen08}, \cite{rabinowitz69}, and
should be compared with Theorem \ref{thm:ptr}. We state the result for the
standard interval $a = -1$, $b = 1$.
\begin{theorem}
  If $f$ is analytic in the region bounded by
  the ellipse with foci $\pm 1$ and major and minor semiaxis lengths
  summing to $\rho > 1$, with $\abs{f(z)} \leq C$, then for any $p \geq 1$,
  \[\abs{\int_{-1}^1 f(x) \, dx - \sum_{j=1}^p f(x_j) w_j} \leq
  \frac{64C}{15\paren{1-\rho^{-2}} \rho^{2p}}.\]
\end{theorem}
As for the PTR, we find exponential convergence for functions which
can be analytically continued to a neighborhood of the interval, with
rate depending on the distance to the closest singularity. We note, however,
that the PTR is superior for functions which are also periodic \cite{hale08}.

\subsection{High-order adaptive integration in 1D}
\label{sec:adapint1d}

Gauss quadrature can be used in a composite
manner by dividing an interval of integration into panels, and using
a $p$-node rule on each panel. The
convergence of such a composite rule is order $2p$, in the sense that if
each panel is split in half to form a new composite rule of the same
order $p$ with double the number of panels, the error decreases
asymptotically by a factor $2^{-2p}$~\cite[Sec.~5.2]{kahaner88}.
This approach allows for the construction of adaptive composite Gauss
quadrature rules, in which panels of different lengths are assembled to
resolve localized features of a function. An automatic and adaptive
algorithm yielding error below a user-specified tolerance $\varepsilon$
is given in Algorithm \ref{alg:adapgau}. We note that many variants
of 1D adaptive integration have been proposed, and several
software packages are available: we refer the reader
to \cite[Sec.~4.7]{press07},
\cite[Sec.~5.6]{kahaner88}, and \cite{piessens12,shampine08} for further
discussion.

\begin{algorithm}[H]
  \caption{Automatic adaptive Gauss quadrature for $f(x)$ on
  $[a,b]$}
  \label{alg:adapgau}
  \textbf{Inputs:} $f(x)$, $a$, $b$, $p$, $\varepsilon$ \\
  \textbf{Output:} An approximation of $I \equiv \int_a^b f(x) \, dx$
  \begin{algorithmic}
    \State Apply $p$-node Gauss quadrature to $f$ on $[a,b]$ to
    obtain $I_0$
    \State Set $c = (a+b)/2$
    \State Apply $p$-node Gauss quadrature to $f$ on $[a,c]$ and $[c,b]$
    to obtain $I_1$ and $I_2$, respectively
    \If{$\abs{I_0 - (I_1+I_2)} \leq \varepsilon$}
      \State Approximate $I$ by $I_1+I_2$
    \Else
      \State \multiline{Call this algorithm on $[a,c]$ and $[c,b]$ with tolerance
      $\varepsilon$; \\ add results to obtain approximation of $I$}
    \EndIf
    \State \textbf{Note:} Since $I_1$ and $I_2$ are used in the
    recursive call to the algorithm (for example $I_0$, as it appears in the
    call on $[a,c]$, is equal to $I_1$, as it appears in the original call on
    $[a,b]$), they can be passed down to avoid repeated work.
  \end{algorithmic}
\end{algorithm}

Algorithm \ref{alg:adapgau} constructs a binary tree of panels on $[a,b]$, and its approximation of $I$ is given by the
composite Gauss quadrature rule on the union of the
leaf-level panels. The 
panels automatically refine towards localized features of $f$. As a comparison
with the PTR, we consider the example
of Fig. \ref{fig:ptrex} with $\eta = 0.01$, for which the PTR of $N = 1000$ points
produces an error between $10^{-4}$ and $10^{-3}$. On the other hand, Algorithm
\ref{alg:adapgau} with $p = 4$ and $\varepsilon = 10^{-4}$ constructs an
adaptive quadrature grid of $256$ points
which produces an error less than $10^{-6}$. For $\eta = 10^{-4}$ with the same
parameters, it constructs a grid of $480$ points which produces an error less than $10^{-7}$.

We note that as for Algorithm \ref{alg:ptr}, it is possible for the
self-consistency check used in Algorithm \ref{alg:adapgau} to
underestimate the true error of $I_0$.
However, because of the high-order accuracy of Gauss quadrature, this self-consistency check tends to
yield an effective error estimate. In the convergence regime, it is
typically a significant overestimate of the error of
$I_1+I_2$, the value ultimately used to approximate the integral $I$.
We also note that since the error
in the approximation of $I$ is bounded by the sum of the errors made on
each leaf-level panel, Algorithm \ref{alg:adapgau} only guarantees $n_\text{pan}
\varepsilon$ rather than $\varepsilon$ error, where $n_\text{pan}$ is
the number of leaf-level panels and $\varepsilon$ bounds the error on
each panel. This can be straightforwardly remedied by using slightly more
sophisticated adaptivity criteria \cite{piessens12}, or
by adding an additional final step to Algorithm \ref{alg:adapgau}:
recompute the integral on all leaf-level panels with
$p$ increased until convergence below $\varepsilon$.

To resolve an isolated feature
which is smooth on the scale $\eta$, Algorithm \ref{alg:adapgau} will
refine panels
dyadically to yield a smallest panel of
width $\OO{\eta}$, giving only $\OO{\log (\eta^{-1})}$ panels (and
hence quadrature nodes) in total.

\section{Two approaches to adaptive Brillouin zone integration} \label{sec:adapintnd}

In this section, we compare two possible generalizations of Algorithm
\ref{alg:adapgau} to 2D and 3D, TAI and IAI, and argue
that the performance of IAI is superior for
integrals of the form \eqref{eq:bzint}.

\subsection{Tree-based adaptive integration}

Algorithm \ref{alg:adapgau} can be straightforwardly generalized to
integrals on square or cubic domains. Rather
than constructing a binary tree of panels and using a Gauss quadrature
rule, one builds a quadtree of squares in two dimensions, or an octree of
cubes in three dimensions, each using a Cartesian product of Gauss
quadrature rules. Several variants of this basic scheme have been proposed and
implemented~\cite{genz80,berntsen91,cubature,hcubature}.

This approach is incompatible with integration on the IBZ, which is a polygon in
2D and a
polytope in 3D. Instead, one can carry out adaptive
integration on a tree of triangles or tetrahedra, as in Refs.~\cite{henk01} and
\cite{assmann16}. In 3D, this requires
partitioning a polytope with a tetrahedral mesh, specifying a splitting
strategy which avoids high-aspect ratio tetrahedra, and building a
high-order quadrature rule for arbitrary tetrahedra. Refs.~\cite{henk01}
and \cite{assmann16}
describe meshing and splitting strategies, but use simple,
low-order quadrature rules.

A more fundamental issue limits the performance of
TAI for integrals of the form
\eqref{eq:bzint}.
Using the term
``cell'' to refer to squares/cubes/triangles/tetrahedra in a tree-based
scheme, we summarize the issue as follows (see the left panel of Fig.~\ref{fig:grids} for
an illustration).
The structure of the integrand necessitates resolving features
of width $\eta$ along the $(d-1)$-dimensional level set
$\det\paren{\omega-H(\vect{k})} = 0$
on which the integrand of \eqref{eq:bzint} with $\eta = 0^+$ diverges
(a curve in 2D and a surface in 3D).
Thus, any tree-based
adaptive scheme refines to cells of diameter $\OO{\eta}$ covering the level set.
Since the length or area of the level set is of unit
magnitude,
and dyadic refinement is also needed in its normal direction,
the total number of cells is $\OO{\log(\eta^{-1})/\eta^{d-1}}$ (in general, for
a level set of dimension $\wb{d}$, the number of cells is
$\OO{\log(\eta^{-1})/\eta^{\wb{d}}}$ regardless of $d$).

\begin{figure}
  \centering
  \includegraphics[width=0.9\linewidth]{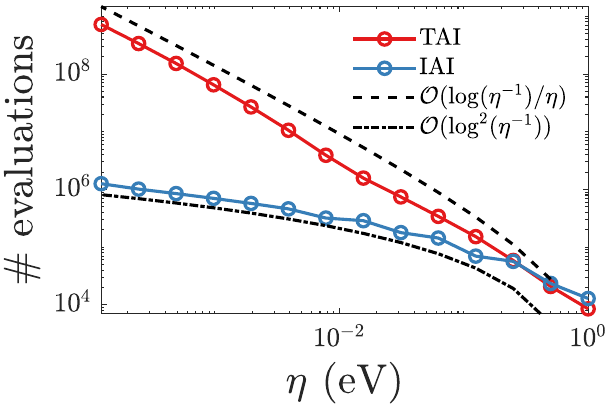}
  \caption{Number of evaluations of the integrand used by TAI and IAI versus $\eta$ for the
  2D example calculation described in
  Fig.~\ref{fig:grids}. Here, we use $\varepsilon = 10^{-5}$, and vary
  $\eta$.}
  \label{fig:taivsiai}
\end{figure}

We verify the
$\OO{\log(\eta^{-1})/\eta^{d-1}}$
scaling of TAI in Fig.~\ref{fig:taivsiai} for the
example shown in Fig.~\ref{fig:grids}. For this calculation, we use the
simple TAI scheme described at the beginning of this
section.
Although it improves on the $\OO{\eta^{-d}}$ scaling of PTR integration,
it is far from optimal, and in practice prevents TAI from accessing the
very small $\eta$ regime.

\subsection{Iterated adaptive integration} \label{sec:iai}

A $d$-dimensional integral can be rewritten as a nested sequence of
1D integrals, each of which can be computed using
adaptive integration. For example, in 2D,
\begin{gather*}
  \int_{a_1}^{b_1} \int_{a_2}^{b_2} dx \, dy \, f(x,y) =
\int_{a_1}^{b_1} dx \, I_2(x), \\
  \mbox{ where} \quad I_2(x) \equiv \int_{a_2}^{b_2} dy \, f(x,y).
\end{gather*}
The integral can therefore be computed by calling
Algorithm \ref{alg:adapgau} on the function $I_2$.
Within this procedure, $I_2(x)$ can be evaluated for fixed $x$ by
calling Algorithm \ref{alg:adapgau} on $f(x,y)$, with $y$
the variable of integration.

For $\bz = [-\pi,\pi]^2$, this IAI method computes \eqref{eq:bzint} as 
\begin{gather*}
  G(\omega) = \int_{-\pi}^{\pi} dk_x \, I_2(k_x), \quad \mbox{ where } \\
  I_2(k_x) = \int_{-\pi}^{\pi} dk_y \, \Tr
\brak{\paren{\omega-H(k_x,k_y) + i \eta}^{-1}}.
\end{gather*}
Fig.~\ref{fig:grids} illustrates the advantage of IAI.
In this case, $H$ is
scalar-valued, so the trace vanishes from the integral. The right panel
shows the inner integral $I_2(k_x)$.
It has only two localized features of width $\OO{\eta}$,
corresponding to the points at which the curve $H(k_x,k_y) = \omega$ has
a tangent aligned with the $k_y$ axis. Thus $I_2(k_x)$ can be obtained
using adaptive integration with a grid of $\OO{\log(\eta^{-1})}$
points (as discussed at the end of Section~\ref{sec:adapint1d}), shown on
the horizontal
axis. Each evaluation of $I_2(k_x)$ requires computing an
integral in $k_y$ with fixed $k_x$. This function also has
at most two localized features of width $\OO{\eta}$, and therefore can
also be adaptively integrated with a grid of
$\OO{\log(\eta^{-1})}$ points. Thus the cost of
computing the integral scales as $\OO{\log^2(\eta^{-1})}$. The resulting
IAI grid is shown in the middle panel.

This scaling is illustrated in Fig.~\ref{fig:taivsiai}.
For $\eta \approx 10^{-4}$, TAI requires over 500 times as many
integrand evaluations as IAI.  The same argument applies for
3D integrals, leading to $\OO{\log^3(\eta^{-1})}$
scaling, so that the advantage of IAI over TAI is even
more significant in this case.

\begin{table}
  \centering
  \setlength{\tabcolsep}{10pt}
  \renewcommand{\arraystretch}{1.5}
  \begin{tabular}{ |c|c|} 
   \hline
   Method & Complexity \\ \hline
    PTR & $\OO{\eta^{-d}}$ \\ 
    TAI & $\OO{\log(\eta^{-1})/\eta^{d-1}}$ \\
    IAI & $\OO{\log^d(\eta^{-1})}$ \\
   \hline
  \end{tabular}
  \caption{Asymptotic complexity as $\eta \to 0^+$ of the number of points required to
  compute $G(\omega)$ in \eqref{eq:bzint}, for fixed $\omega$, to
  a given accuracy, using the three methods described above.}
  \label{tab:scalings}
\end{table}

We summarize the asymptotic scalings of the three schemes we have
discussed in Table \ref{tab:scalings}. Of the two adaptive methods, IAI
is almost always preferable, and we therefore do not consider TAI in our
numerical tests.

\subsection*{Rapid evaluation of $H(\texorpdfstring{\vect{k}}{k})$ in iterated
adaptive integration}

Since the nearly-singular set
varies with $\omega$, adaptive methods must
evaluate $H(\vect{k})$ on the fly for each fixed $\omega$, a
computational bottleneck in certain cases. As discussed in Section~\ref{sec:ptr}, we evaluate $H(\vect{k})$ using its truncated Fourier series
representation \eqref{eq:fstrunc}.
This can be done efficiently in IAI by evaluating the series in a
sequential manner. We demonstrate this in the 2D case, writing $\vect{R} = (m,n)$
and $H_{mn} \equiv H_{\vect{R}}$ in \eqref{eq:fstrunc} to simplify notation:
\[H(k_x,k_y) = \frac{1}{(2\pi)^2} \sum_{m,n=-M/2}^{M/2-1}
e^{i(m k_x+n k_y)} H_{mn}.\]
At each step of the outer integration, $k_x$ is fixed, and
an inner integration in $k_y$ is carried out. Before the
inner integration, we can precompute and store
the $M$ numbers
\begin{equation} \label{eq:1dfscoeffs}
  H_n(k_x) \equiv \frac{1}{(2 \pi)^2} \!\!
\sum_{m=-M/2}^{M/2-1} e^{i m k_x} H_{mn}.
\end{equation}
Then $H(k_x,k_y)$ can be evaluated at $k_y$ using
\begin{equation} \label{eq:1dfs}
  H(k_x,k_y) \;= \sum_{n=-M/2}^{M/2-1} e^{i n k_y} H_n(k_x).
\end{equation}
In this way, the cost of evaluating the part of the Fourier series
corresponding to the outer integration variable is amortized over all inner
quadrature nodes, and becomes negligible.
The cost of evaluating
$H(\vect{k})$ in IAI is thus dominated by that of evaluating the 1D Fourier
series \eqref{eq:1dfs}.
This can be done efficiently by using
the simple recurrence $e^{\pm i n k_y} = e^{\pm i k_y} e^{\pm i (n-1)
k_y}$ to compute the exponentials.

\subsection*{Integration on the irreducible Brillouin zone}

Given an explicit Cartesian parameterization of the IBZ, IAI can
be used directly on the IBZ by making the straightforward domain replacement
\[\intbz dk_x \, dk_y \, dk_z \, \gets \int_{a_1}^{b_1} dk_x \,
\int_{a_2(k_x)}^{b_2(k_x)} dk_y \, \int_{a_3(k_x,k_y)}^{b_3(k_x,k_y)}
dk_z\]
in 3D, and similarly in 2D.
Since the IBZ is a convex polygon/polytope, $a_j$ and $b_j$
are piecewise affine. For example,
for the simple cubic BZ
$[-\pi,\pi]^3$, the IBZ is a tetrahedron
with $1/48$th of its volume given by $a_1 = 0$,
$b_1 = \pi$, $a_2(k_x) = k_x$, $b_2(k_x) = \pi$, $a_3(k_x,k_y) = 0$,
$b_3(k_x,k_y) = k_x$.

In practice, the IBZ may only be described by a collection of point
symmetries rather than an explicit parameterization.
Ref.~\cite{jorgensen22} describes an algorithm and
software package which determines the
IBZ as a convex hull characterized by its faces,
edges, and vertices. Given this convex hull, it is straightforward to
determine $a_1$ and $b_1$, and to evaluate the functions
$a_2$, $b_2$, $a_3$, and $b_3$.

We note that the presence of vertices in the IBZ can introduce isolated points
of non-differentiability in the integrand, which interferes with high-order
convergence. This is straightforwardly remedied by splitting the interval of
integration into subintervals determined by these points.

\section{Frequency domain adaptivity} \label{sec:adapinterp}

When $\eta$ is small, $G(\omega)$ itself may
develop localized features.
A common example is a Van Hove
singularity in the
$\eta\to0^+$ limit of $A(\omega)$, defined in \eqref{eq:dos}; when $\eta > 0$,
singularities become regularized on a length scale $\OO{\eta}$.
Automatic adaptive polynomial interpolation is a standard algorithm which can be
used to resolve such features, producing an approximation of
a function $f(x)$ on an interval $[a,b]$ accurate
to a user-specified error tolerance.
It operates in a similar
manner to Algorithm \ref{alg:adapgau}. A polynomial interpolant of
degree $p-1$ is built on $[a,b]$ from samples of $f(x)$ at $p$
nodes. The panel is accepted if the estimated error of the interpolant is below a user-specified
tolerance, and split otherwise. If the panel is split,
polynomial interpolants of degree $p-1$ are similarly constructed on
each of the two resulting panels, and the process is repeated. The result is a
piecewise degree $p-1$ polynomial interpolant on a collection of panels which are
automatically adapted to the localized features of $f(x)$. The
interpolant can be evaluated rapidly at any given point $x$ by descending
the binary tree produced by the procedure to
identify the leaf-level panel containing $x$, and then evaluating the
polynomial interpolant corresponding to that panel. The interpolant is uniformly
accurate to the specified tolerance.

A complete description of the procedure requires identifying specific
polynomial interpolation and error estimation methods. For the former, we use barycentric interpolation at
Chebyshev nodes, which is fast, stable, well-conditioned, and spectrally
accurate \cite{berrut04,trefethen19}.
There are several possible error estimation strategies. A fairly robust
approach is to compare an interpolant on $[a,b]$ to interpolants on
$[a,(a+b)/2]$ and $[(a+b)/2,b]$ on a dense grid of points, and estimate
the error as the maximum absolute discrepancy. A slightly less robust
but commonly used approach is to use the sum of the absolute values of
the last few coefficients of the Chebyshev expansion on a panel as an
error estimate. This method is often effective in practice, and requires
fewer evaluations of $f(x)$, so we use it for the calculations
carried out in the next section.

This automatic adaptive Chebyshev interpolation procedure can be applied
to obtain a piecewise polynomial representation of the spectral function
$A(\omega)$ on a given interval $[\omega_{\min},\omega_{\max}]$. The
interpolants on each panel are constructed from samples of $A(\omega)$
at particular values of $\omega$, and these samples can be obtained
using either the automatic PTR method described in
Section \ref{sec:autoptr}, or IAI.
To resolve localized features at a scale
$\OO{\eta}$, the adaptive
interpolant requires $\OO{\log (\eta^{-1})}$
evaluations of
$A(\omega)$, each of which requires computing a BZ integral. The expected
scaling of the spectral function calculation as
$\eta \to 0^+$ is therefore $\OO{\log(\eta^{-1})/\eta^3}$
using the automatic PTR method and $\OO{\log^4(\eta^{-1})}$
using IAI.

\section{Example: spectral function of
strontium vanadate} \label{sec:results}

We demonstrate the automatic IAI and PTR methods, as well as the frequency
adaptivity procedure, by calculating the
spectral function for 
the correlated metal strontium vanadate, SrVO$_3$.
Its cubic structure (with space group Pm$\bar3$m (221)) and the isolated
set of three $t_{2g}$-derived low-energy states results in a simple
tight-binding-like electronic structure of three degenerate orbitals,
yielding $3 \times 3$ matrices $H(\vect{k})$.
SrVO$_3$ is known to exhibit Fermi liquid behavior~\cite{Dougier/Fan/Goodenough:1975}, with the typical scattering rate scaling given in \eqref{eq:fermiliquid}.
For now, we consider a scattering rate $\eta$ which is constant
in frequency, but expect our
algorithms to
be particularly useful for frequency-dependent self-energies.

We construct the tight-binding Hamiltonian $H_{\vect{R}}$ in \eqref{eq:fstrunc} from
the low-lying $t_{2g}$ manifold by projecting the ab initio
Hamiltonian onto the three partially-filled vanadium $3d$ orbitals
using the \textsc{Wannier90} code~\cite{wannier90}. Using the planewave-based
\textsc{Quantum~ESPRESSO} package~\cite{Giannozzi_et_al:2009}, we
evaluate the ab initio
Hamiltonian at $M=10$ equispaced grid points per
dimension to obtain the samples $H(q_m)$ referred to in
Section~\ref{subsec:ptrbz}.
We use default parameters for the ground state calculation,
with lattice parameter $a = 3.859$ \AA, and the Perdew–Burke–Ernzerhof functional~\cite{Perdew/Burke/Ernzerhof:1996} with standard scalar-relativistic ultrasoft
pseudopotentials~\cite{Garrity_et_al:2014}.

We compute the spectral function as defined in \eqref{eq:dos} and
\eqref{eq:bzint} with $d = 3$. We carry out the integration over the IBZ, which
contains 48 point symmetries,
using the methods described in Appendix \ref{app:ibz} and Section~\ref{sec:iai} for
the PTR and IAI, respectively. We set $\varepsilon = 10^{-5}$
for both the
automatic PTR method and IAI, and use $p = 4$ Gauss quadrature nodes
per panel for IAI to obtain eighth-order accurate quadratures.
All calculations are performed on a single core of a workstation with an
Intel Xeon Gold 6244 processor and 256 GB of RAM.

\begin{figure}
  \centering
  \includegraphics[width=0.9\linewidth]{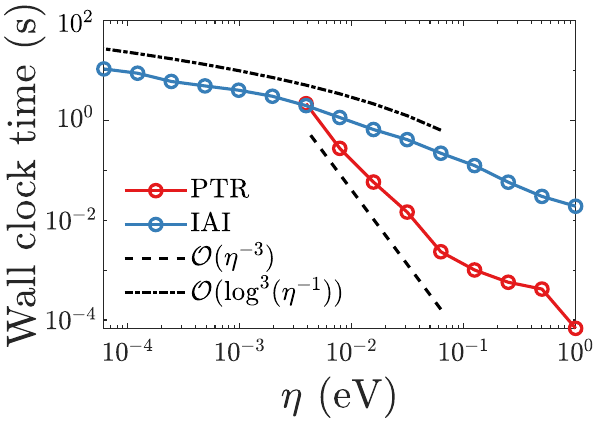}
  \caption{Wall clock time versus $\eta$ for the calculation of the spectral function $A(\omega)$ of SrVO$_3$ at $\omega = 0$ eV, comparing the automatic PTR and
  IAI methods. For the automatic PTR, we only include the costs of summation for the
  final values of $N = \xtrial{N}$ and $N = \xtest{N}$ chosen by the algorithm,
  since other costs, particularly that of evaluating $H(\vect{k})$ on the PTR
  grid, would be amortized over many frequencies in a calculation of the full
  spectral function.}
  \label{fig:srvo3times_onefreq}
\end{figure}

In Fig. \ref{fig:srvo3times_onefreq}, we show wall clock timings for the
calculation of the spectral function at the Fermi level $\omega =0$ eV,
for several values of $\eta$ down to the sub-meV scale. We observe the expected asymptotic
scalings. 
The automatic PTR method could not be used for $\eta \lesssim 4$ meV due to memory constraints, specifically that
of storing $H(\vect{k})$ on a fine equispaced grid.
$H(\vect{k})$ could be evaluated on the fly as in IAI, rather
than computing and storing it once, but when $A(\omega)$ must be
computed at more than one value of $\omega$, this
significantly increases the computational cost.
By contrast, IAI has modest memory
requirements which do not increase significantly as $\eta$ is decreased.

\begin{figure*}
  \centering
  \includegraphics[width=0.8\linewidth]{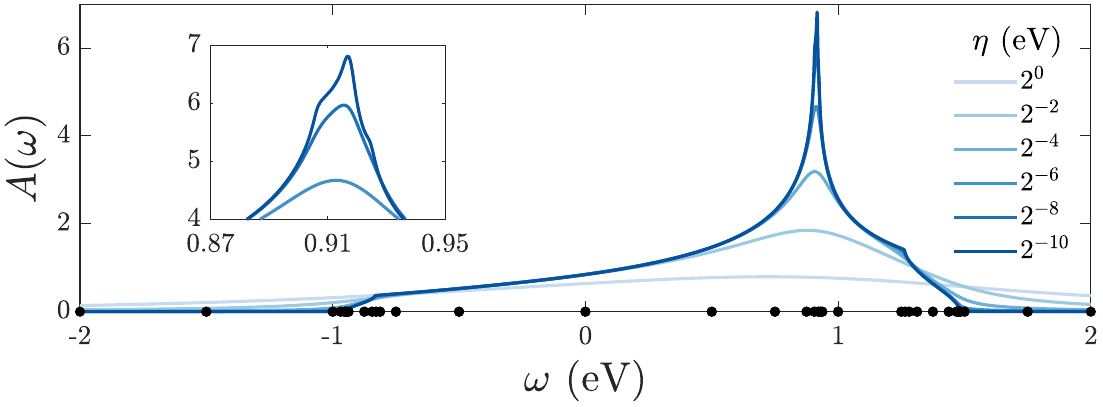}
  \caption{
    Spectral function $A(\omega)$ of SrVO$_3$ with values
  of $\eta$ varying from $1$ eV to $2^{-10}$ eV $\approx 1$ meV. The inset shows details of the region near an approximate
  singularity, regularized on the scale $\eta$.
  }
  \label{fig:dos_srvo3}
\end{figure*}

We next compute the full spectral function $A(\omega)$ using the adaptive
interpolation method described in Section~\ref{sec:adapinterp}. We use a
tolerance $10^{-4} = 10 \varepsilon$, and a Chebyshev
interpolant of $q = 16$ nodes on each panel.
The result is shown in Fig.~\ref{fig:dos_srvo3} for
several values of $\eta$, along with the final adaptive frequency grid
obtained by the automatic algorithm for $\eta \approx 1$
meV.
We observe multiple Van Hove band edges, 
as well as a
stronger feature at $\omega \approx 0.9$ eV corresponding to flat bands.
The adaptive interpolation algorithm refines into all localized features, as expected. The maximum error, measured against a converged reference
for each $\eta$ obtained by repeating all calculations with $\varepsilon
= 10^{-8}$, is at most $7.4 \times 10^{-5}$ using the automatic PTR
method and $4.7 \times 10^{-4}$ using the IAI method.

We plot wall clock timings for both methods in Fig.
\ref{fig:srvo3times}, and observe the expected scaling with $\eta$.
IAI is faster than the
automatic PTR method for $\eta \lesssim 15$ meV, and an optimal scheme
would switch from one method to the other at this point. 
Table \ref{tab:nptr} shows the number of PTR grid points per dimension
used by the automatic PTR algorithm to obtain the specified
accuracy for each $\eta$.
The automatic PTR
method cannot be used for $\eta \lesssim 4$ meV due to memory
constraints, and at $\eta \approx 4$ meV, the finest required grid uses $N =
2714$ points per dimension.
For $\eta \approx 1$ meV, computing $A(\omega)$ using IAI takes just under
25 minutes. Extrapolating from the data shown, and ignoring memory
limitations, the automatic PTR method would have taken
over $4.5$ days to complete the same calculation---over $250$ times as
long as IAI---and would have required over $N = 10\,000$ points per
dimension. Although the IAI calculation used a single core,
several straightforward options for parallelization are available.
For example, a simple parallelization over frequency points would yield this
result in a few minutes on a modern workstation.

\begin{figure}[b]
  \centering
  \includegraphics[width=0.9\linewidth]{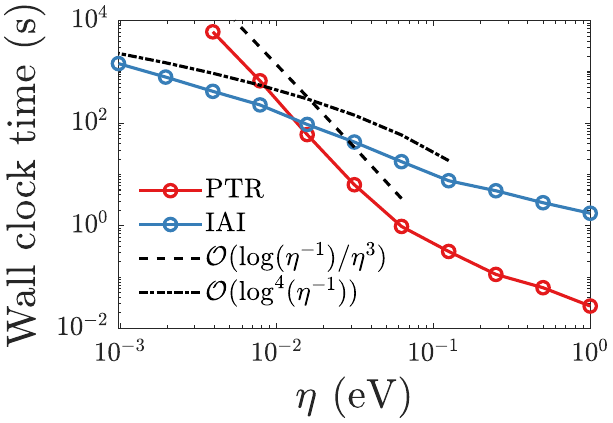}
  \caption{Wall clock time versus $\eta$ for the calculation of the spectral function $A(\omega)$ of SrVO$_3$, using
  adaptive frequency interpolation with both the automatic PTR
  method and IAI.}
  \label{fig:srvo3times}
\end{figure}

\begin{table}
  \centering
  \setlength{\tabcolsep}{4.0pt}
  \renewcommand{\arraystretch}{1.4}
  \begin{tabular}{|c|c|c|c|c|c|c|c|c|c|} 
   \hline
    $\eta$ (eV) & $2^0$ & $2^{-1}$ & $2^{-2}$ & $2^{-3}$ & $2^{-4}$ & $2^{-5}$ & $2^{-6}$ & $2^{-7}$ & $2^{-8}$ \\ \hline
    $N$    & 65 & 112 & 124 & 148 & 196 & 340 & 678 & 1356 & 2714 \\
   \hline
  \end{tabular}
  \caption{Number $N$ of PTR grid points per dimension used for the
  integration error tolerance
  $\varepsilon = 10^{-5}$
  to compute the spectral function of SrVO$_3$ shown in Fig.
  \ref{fig:dos_srvo3}.}
  \label{tab:nptr}
\end{table}

\section{Conclusion} \label{sec:conclusion}

We have presented automatic high-order accurate methods for BZ integration with positive broadening, and adaptive frequency sampling. Our algorithms are
straightforward to implement and deliver results accurate to a user-specified
error tolerance $\varepsilon$. For integration, we recommend the automatic PTR method
when $\eta$ is not too small, and IAI otherwise. Several
possible heuristics could be used to choose between the two methods based on
$\eta$ and $\varepsilon$, and we will investigate these in our
future work.
For adaptive
calculations, we have shown that IAI is preferable to TAI, because of its superior asymptotic scaling as $\eta \to 0^+$ and
algorithmic efficiency, as well as its simplicty of implementation for
IBZ calculations. In particular, for integrals of the form
\eqref{eq:bzint}, IAI has the mild asymptotic complexity
$\OO{\log^d(1/\eta)}$.

Although for concreteness we have restricted our discussion to BZ integrals for
Green's functions of the form \eqref{eq:bzint}, and have demonstrated
the performance of
our algorithms in calculations of the spectral function with constant,
diagonal
self-energy $\Sigma=-i \eta$, our framework is fairly general. In particular,
the case of Green's functions involving frequency-dependent and non-diagonal
self-energies is straightforward, and other possible applications include BZ integrals arising in
response functions. The generalization to momentum-dependent
self-energies $\Sigma(\vect{k},\omega)$ requires that $\Sigma$
can be evaluated rapidly at arbitrary points $(\vect{k},\omega)$. Although the given representation of $\Sigma$ may be
expensive to evaluate, it can be replaced by an approximant which can
be evaluated quickly (e.g., a piecewise polynomial interpolant) in a
precomputation. The proper approximant depends on the underlying
structure of $\Sigma$, and is therefore problem-dependent.

\acknowledgments

We thank J. Bonini, S. Ponc\'e, J.-M. Lihm, M. Ghim, C.-H. Park, A. Togo, A.
Hampel, C. Dreyer, A. Levitt, and E. Letournel for helpful discussions. We also
thank S. Tsirkin for pointing out the algorithm described in Appendix
\ref{app:autoptralt}, and for several other useful suggestions on the manuscript.
The Flatiron Institute is a division of the Simons Foundation. 

\appendix

\section{Convergence of the periodic trapezoidal rule for Brillouin zone
integrals} \label{app:ptrtheory}

In this Appendix, we develop convergence results for the PTR which extend the
discussion in Section~\ref{sec:ptr}.
Given analytic $H(\vect{k})$, we would like a practical
lower bound on $a$, the
convergence rate appearing in an error estimate $e_N \le Ce^{-aN}$
for the $d$-dimensional PTR with $N$ points per dimension (as in Theorem
\ref{thm:ptr} for 1D).
Such bounds can be used to guide the selection of $\Delta N$ in
Algorithm \ref{alg:ptr} in order to make the error estimate
as robust as possible.
Algorithm \ref{alg:ptr}
estimates the error of the $N^d$-point PTR as $\abs{e_N - e_{N+\Delta N}}$.
If $e_{N+\Delta N}$ is
smaller than $e_N$ by a factor $\gamma\gg1$,
then this error estimate is accurate with high probability. Thus we
must choose $\Delta N$ such that
\begin{equation} \label{eq:dnestimate}
  \gamma \leq e_N/e_{N + \Delta N} \approx e^{a \Delta N}
  \implies \Delta N \geq \log(\gamma)/a.
\end{equation}

To justify such an exponential error estimate,
we first
generalize Theorem \ref{thm:ptr} to $d>1$.
We then give a rigorous lower bound on its rate $a$, in terms of the width of
the strip of analyticity, in each dimension, of the integrand of
\eqref{eq:bzint}. The bound takes the form $a
\geq C \eta$, where $C$ depends on an upper bound of
$\nabla_{\vect{k}} H(\vect{k})$, in a norm to be specified for the case
of scalar-valued $H(\vect{k})$, and conjectured otherwise.
If this bound is on the order of 1,
then \eqref{eq:dnestimate} with $\gamma = 10$ gives $\Delta N \gtrsim
2.3/\eta$. We have checked that this is the case for the example
considered in Section~\ref{sec:results}, and therefore use $\Delta N =
2.3/\eta$ for our numerical experiments.

\subsection*{Generalization of Theorem \ref{thm:ptr} to dimension $d$}

For simplicity, we state the result for 2D, but the same technique
gives an analogous result in 3D.
\begin{theorem}\label{thm:ptr2d}
  Let $f(x,y)$ be doubly $2\pi$-periodic. Suppose
  that for each fixed $y=y_0$, $f(x,y_0)$ is analytic in the
  strip $\abs{\Im x} < a_1$, with $\abs{f(x,y_0)} \leq C_1$, and
  for each fixed $x=x_0$, $f(x_0,y)$ is analytic in the
  strip $\abs{\Im y} < a_2$, with $\abs{f(x_0,y)} \leq C_2$.
  Then for any $N_1,N_2 \geq 1$,
  \begin{widetext}
    \[\abs{\int_{[0,2\pi]^2} f(x,y) \, dx \, dy 
    - \frac{\paren{2\pi}^2}{N_1 N_2} \sum_{n_1,n_2=0}^{N_1-1,N_2-1}
   f\paren{\frac{2\pi
  n_1}{N_1},\frac{2\pi n_2}{N_2}}} \leq
    8 \pi^2 \paren{\frac{C_1}{e^{a_1 N_1}-1} + \frac{C_2}{e^{a_2 N_2}-1}}.\]
  \end{widetext}
\end{theorem}
\begin{proof}
  The triangle inequality gives
\begin{widetext}
  \begin{multline*}
    \abs{\int_{[0,2\pi]^2} dx \, dy \, f(x,y)
      - \frac{\paren{2\pi}^2}{N_1 N_2} \sum_{n_1,n_2=0}^{N_1-1,N_2-1}
       f\paren{\frac{2\pi
      n_1}{N_1},\frac{2\pi n_2}{N_2}}}
      \leq \abs{\int_0^{2\pi} dx \, \brak{\int_0^{2\pi} dy \, f(x,y) 
      - \frac{2\pi}{N_2} \sum_{n_2=0}^{N_2-1} f\paren{x,\frac{2\pi
      n_2}{N_2}}}} \\
      + \abs{\frac{2\pi}{N_2} \sum_{n_2=0}^{N_2-1} \int_0^{2\pi} dx \, f\paren{x,\frac{2\pi
      n_2}{N_2}} 
      - \frac{\paren{2\pi}^2}{N_1 N_2} \sum_{n_1,n_2=0}^{N_1-1,N_2-1}
      f\paren{\frac{2\pi n_1}{N_1},\frac{2\pi n_2}{N_2}}} \equiv T_1 +
      T_2.
  \end{multline*}
\end{widetext}
  Fixing $x$ in the inner integral of $T_1$ and applying Theorem~\ref{thm:ptr} gives
  \[T_1 \leq \int_0^{2\pi} dx \, \frac{4 \pi C_2}{e^{a_2 N_2}-1} =
  \frac{8 \pi^2 C_2}{e^{a_2 N_2}-1}\]
  for any $N_2 \geq 1$. We then have
\begin{widetext}
  \[T_2 = \abs{\frac{2\pi}{N_2} \sum_{n_2=0}^{N_2-1} \brak{\int_0^{2\pi} dx \, f\paren{x,\frac{2\pi
      n_2}{N_2}} 
      - \frac{2\pi}{N_1} \sum_{n_1=0}^{N_1-1}
      f\paren{\frac{2\pi n_1}{N_1},\frac{2\pi n_2}{N_2}}}} \leq 
      \frac{2\pi}{N_2} \sum_{n_2=0}^{N_2-1} \frac{4 \pi C_1}{e^{a_1
      N_1}-1} = \frac{8 \pi^2 C_1}{e^{a_1 N_1}-1}\]
\end{widetext}
  for any $N_1 \geq 1$, where we have fixed $y = \frac{2 \pi n_2}{N_2}$ in the inner integral and
  applied Theorem \ref{thm:ptr}.
\end{proof}
Setting $N_1 = N_2$ and taking $a = \min(a_1,a_2)$, $C = \max(C_1,C_2)$, we
obtain the error bound $\frac{16 \pi^2 C}{e^{a N}-1}$, which is of the
form required to justify the argument underlying
\eqref{eq:dnestimate}. In $d$ dimensions,
the bound is $\frac{2d\paren{2\pi}^d C}{e^{aN}-1}$.

\subsection*{Lower bound on the analytic strip width $a$}

We begin with the simplest case of \eqref{eq:bzint}: $d = 1$ with $H(k)$
scalar-valued. The following argument gives some intuition for the
discussion. Let $k_0\in\RR$ be such that $H(k_0)
= \omega$. If $\eta$ is small, we expect a pole of the integrand near $k_0$ in the complex
plane. A Taylor expansion about $k = k_0$ approximates the pole location
as
\[k = k_0 + \frac{i \eta}{H'(k_0)}.\]
We therefore expect that $\paren{\omega-H(k)+i\eta}^{-1}$ is
analytic in a strip $\abs{\Im k} < a$ of width
\begin{equation} \label{eq:a1dtaylor}
  a \gtrsim \frac{\eta}{\max_k \abs{H'(k)}}.
\end{equation}
In 2D, in order to estimate a lower bound on the analytic strip widths mentioned in Theorem \ref{thm:ptr2d}, we can repeat the argument with $k \gets k_x$ for each
fixed $k_y$, and similarly with $k \gets k_y$ for each fixed $k_x$. We
obtain
\[a_1 \gtrsim \frac{\eta}{\max_{\vect{k}}
\abs{\pd{H}{k_x}(k_x,k_y)}}, \quad a_2 \gtrsim \frac{\eta}{\max_{\vect{k}}
\abs{\pd{H}{k_y}(k_x,k_y)}}.\]
We can follow a similar procedure in 3D.

The following result clarifies this intuition more rigorously,
starting with 1D. We absorb
$\omega$ into $H$ as it does not affect the result.

\begin{lem}
  \label{lem:strip1d}
  Let $H$ be a $2\pi$-periodic real analytic
  function on $\RR$, and let $\eta > 0$.
  Suppose $a>0$ is such that
  i) $H$ can be analytically continued throughout
  the closed strip $|\Im z| \le a$, and
  ii) $\max_{|\Im z| \le a} |H'(z)| < \eta/a$.
  Then $(H(z) + i\eta)^{-1}$ is analytic and bounded in this closed strip.
\end{lem}
\begin{proof}
  Since $H(z)$ is analytic in the strip $\abs{\Im z} \leq a$, $(H(z) +
  i\eta)^{-1}$ is also analytic
  there except possibly at isolated poles. Note that the case of
  constant $H\equiv-i\eta$ is not possible since $H$ is real on $\RR$.
  Suppose $z_p$ is such a pole, i.e.,
  $H(z_p)=-i\eta$ with $\abs{\Im z_p} \leq a$. Then
  \[-i\eta - H(x_p) = \int_{x_p}^{z_p} H'(z) \, dz\]
  with $x_p = \Re z_p$.
  Bounding the magnitudes of both sides gives
  $\eta \le |i\eta + H(x_p)| \le a \max_{|\Im z| \le a} |H'(z)|$,
  since $H(x_p) \in \RR$.
  This contradicts the hypothesis ii).
  Therefore $(H(z) + i\eta)^{-1}$ is analytic and bounded in the closed strip.
\end{proof}

This result can be used to obtain a practical lower bound for the rate
of exponential decay provided by Theorem \ref{thm:ptr} in the present
setting, in which $H(k)$ is a truncated Fourier series, and therefore
entire. We define $M(\alpha) \equiv \max_{\abs{\Im z} \leq \alpha}
\abs{H'(z)}$, a continuous, non-decreasing function of $\alpha \geq 0$ which
may be readily computed. Using the bisection algorithm, we can then find a
height $\alpha^\ast>0$ as large as possible such that $\alpha^\ast \,
M(\alpha^\ast) < \eta$.
The hypotheses of Lemma \ref{lem:strip1d} are satisfied with $a$
replaced by $\alpha^\ast$, so we find that the exponential rate of decay $a$
appearing in
Theorem \ref{thm:ptr} satisfies 
\[a \geq \frac{\eta}{M(\alpha^\ast)} = \frac{\eta}{\max_{\abs{\Im z}
\leq \alpha^\ast} \abs{H'(z)}}.\]
Since $M(\alpha^\ast) \to \max_k \abs{H'(k)}$ as $\eta \to 0^+$, we see that \eqref{eq:a1dtaylor}
is a reasonable approximation in this limit.

In 2D, we can define $M_1(\alpha_1) \equiv \max_{k_y} \max_{\abs{\Im z} \leq
\alpha_1} \abs{\pd{H}{k_x}(z,k_y)}$, and $\alpha_1^\ast > 0$ a height as large as
possible with $\alpha_1^\ast M_1(\alpha_1^\ast) < \eta$. $M_2(\alpha_2)$
and $\alpha_2^\ast$ can be defined similarly by switching the roles of $k_x$ and
$k_y$. We then obtain the lower bounds
\begin{align*}
  a_1 &\geq \frac{\eta}{M_1(\alpha_1^\ast)} = \frac{\eta}{\max_{k_y}
\max_{\abs{\Im z} \leq \alpha_1^\ast}
\abs{\pd{H}{k_x}(z,k_y)}}, \\ 
  a_2 &\geq \frac{\eta}{M_2(\alpha_2^\ast)} = \frac{\eta}{\max_{k_x}
\max_{\abs{\Im z} \leq \alpha_2^\ast}
\abs{\pd{H}{k_y}(k_x,z)}},
\end{align*}
which can then be inserted into Theorem \ref{thm:ptr2d}.

Rigorous convergence rate bounds
for the case of matrix-valued $H(\vect{k})$ are more challenging, and we
leave them for future work. For $d = 1$, we conjecture that the
arguments above will be valid with $\abs{H'(z)}$ replaced by
$\norm{H'(z)}_2$. The $d>1$ case might then
be treated in a similar manner as above.

\section{Symmetrized periodic trapezoidal rule in the irreducible Brillouin zone}
\label{app:ibz}

Let $S_1,\ldots,S_p$ be the orthogonal matrices representing the point
symmetries of the lattice, so that $H(S_j\vect{k}) = H(\vect{k})$. We define
$\ibz$, an irreducible wedge of the BZ, as the
closed set such that $\bz = \cup_{j=1}^p S_j \, \ibz$ and
$\cap_{j=1}^p \paren{S_j \, \ibz}^{\mathrm{o}} = \varnothing$. Denoting the integrand of
\eqref{eq:bzint} by $f\paren{H(\vect{k})}$, and making the
change of variables $\vect{k} \gets S_i \vect{k}$, we can
calculate \eqref{eq:bzint} as an integral over $\ibz$:
\begin{equation} \label{eq:ibzint}
\begin{multlined}
  \intbz d^d\vect{k} \, f\paren{H(\vect{k})} = \sum_{j=1}^p \int_{S_j
  \ibz} d^d \vect{k} \, f\paren{H(\vect{k})} \\
  = \sum_{j=1}^p \int_{\ibz} d^d \vect{k} \, f\paren{H(S_j \vect{k})}
  = p \int_{\ibz} d^d \vect{k} \, f\paren{H(\vect{k})}.
\end{multlined}
\end{equation}

The PTR cannot be applied directly to the integral over $\ibz$, since it is
not a regular domain on which $H(\vect{k})$ is periodic. However, under
some constraints, the $d$-dimensional PTR can be modified to take
advantage of the point symmetries. In
particular, if we take the
number of grid points in each dimension to be equal (e.g. $N_1 = N_2$
in \eqref{eq:ptr2d}), then the PTR sum can be converted exactly
into a sum over only the grid
points of the PTR lying in the IBZ, suitably reweighted.

Indeed, in this case, the PTR grid respects the point symmetries.
Let $S$ be a particular point symmetry.
Since $S$ maps any reciprocal lattice vector to another
reciprocal lattice vector, we have $SB = BM$ for some $M \in \ZZ^{d \times d}$.
Consider a point $\vect{g} = B\vect{n}/N$ of the $N^d$-point PTR lying
in $\ibz$, $\vect{n} \in \ZZ^d$ with $0 \leq n_1,\ldots,n_d \leq N-1$.
We have $S\vect{g} = BM\vect{n}/N \in S\ibz$, an integer linear
combination of the vectors $B/N$, i.e., a PTR grid point in $S\ibz$.
Therefore all PTR grid points in $\ibz$ map
to PTR grid points in $S \ibz$, and all PTR grid points in $S \ibz$ are the
image under $S$ of a grid point in $\ibz$, as can be seen by repeating
the argument on $S \ibz$ with the symmetry $S^{-1}$. As a consequence, one can
reduce the full PTR sum to a sum over the grid points $\vect{g} \in \ibz$, weighted
by $w_{\vect{g}} = \abs{\cup_{i=1}^p S_i \vect{g}}$, the number of
distinct images of $\vect{g}$ in the
BZ under the point symmetries. This number is equal to $p$ for interior
points of the IBZ, and less than $p$ for boundary points.

Letting $K = \{B\vect{n} \, | \, 0 \leq n_1,\ldots,n_d  \leq N-1,
B\vect{n} \in \ibz\}$ be the collection of PTR grid points in $W$ (or,
more generally, a minimal sufficient set of grid points which may not
all fall in a single irreducible wedge), 
we obtain an analogue of \eqref{eq:ibzint} for the $N^d$-point PTR: 
\[\frac{\abs{\det B}}{N^d} \sum_{n_1,\ldots,n_d=0}^{N-1} f\paren{H(B\vect{n})} =
\frac{\abs{\det B}}{N^d} \sum_{\vect{g} \in K} w_{\vect{g}}
f\paren{H(\vect{g})}.\]
The set $K$ and weights $w_{\vect{g}}$ can be computed by the following
algorithm, which does not require a parameterization of the IBZ, but
only the point symmetries. 
\begin{algorithm}[H]
  \caption{Determine grid points, weights of symmetrized
  $N^d$-point PTR}
  \label{alg:ptrsym}
  \textbf{Inputs:} $N$, $S_1,\ldots,S_p$ \\
  \textbf{Outputs:} $K$, $w_{\vect{g}}$ 
  \begin{algorithmic}
    \State Initialize $K$ as the full $N^d$-point PTR grid
    \Loop{ over remaining points $\vect{g} \in K$}
      \State $w_{\vect{g}} = 1$
      \For{$j=2,\ldots,p$}
        \If{$S_j \vect{g} \in K \wedge S_j \vect{g} \neq \vect{g}$}
          \State $K \gets K \, \backslash \, \{S_j \vect{g}\}$
          \State $w_{\vect{g}} \gets w_{\vect{g}} + 1$
        \EndIf
      \EndFor
    \EndLoop
    \State \textbf{Note:} We assume here that $S_1$ is the identity map.
  \end{algorithmic}
\end{algorithm}

In some cases, the PTR can be symmetrized under weaker
constraints than that considered here. For example, the 2D PTR respects
all symmetries of any non-square rectangular lattice even if a different number of grid points is chosen
in each dimension. It may also be possible to take advantage of partial
symmetries. For example, the PTR respects half of the symmetries of the square lattice if a different
number of grid points is used in each dimension, and can therefore be
symmetrized to a domain consisting of two irreducible wedges. We do not attempt to classify
all such possibilities.

\section{Efficient evaluation of \texorpdfstring{$H(\vect{k})$}{H(k)} for periodic trapezoidal
rule} \label{app:ptrheval}

In our PTR integration procedure, we compute $H(\vect{k})$ at all grid
points before summing.
This could be done using a zero-padded $N^d$-point FFT, but that is 
typically not the most efficient
approach for the present case in which $N \gg M$. Instead, we evaluate the Fourier series directly
using the iterated summation approach
summarized in \eqref{eq:1dfscoeffs} and \eqref{eq:1dfs}: we loop over grid points
dimension-by-dimension, precomputing Fourier series coefficients for fixed
components of $\vect{k}$ so that that dominant cost is evaluating a truncated
1D Fourier series of $M$ terms at each point. 
For a full BZ calculation, the total cost scales as $\OO{M N^d}$.
For the symmetrized PTR on the IBZ, described in Appendix \ref{app:ibz}, some care must be taken to implement iterated
summation only over the points in $K$. In particular, for each index
$n_i$ in an outer sum, one must keep track of the the range of indices
$n_{i+1},\ldots,n_d$ in each inner sum corresponding to points included
in $K$. The cost of the resulting algorithm scales as $\OO{M \abs{K}} =
\OO{M N^d/p}$. 

An alternative approach is the pruned FFT algorithm \cite{markel71}, described
in the context of BZ integration in Ref.~\cite{tsirkin21}. The pruned FFT, which
has $\OO{N^d \log M}$ complexity, reconstructs the $N^d$-point Fourier transform
from a collection of rephased $M^d$-point FFTs. When $M$ is small, it is not
clear whether or not this reduced complexity translates to a performance improvement
over direct evaluation of the Fourier series. We therefore defer a comparison
with the pruned FFT to future work.

\section{Alternative automatic PTR algorithm for fixed frequency grids} \label{app:autoptralt}

If a collection $\omega_j$ of frequencies at which to evaluate $G(\omega)$ is
specified in advance, there is an alternative to the automatic PTR algorithm
described in Section \ref{sec:autoptr} which has two advantages: (1) it avoids
storing $H(\vect{k})$ on a PTR grid, and (2) it avoids the possibility of computing
$G(\omega_j)$ on a denser PTR grid than is necessary for a given frequency
$\omega_j$. This procedure is described in Algorithm \ref{alg:ptralt}. We denote the integrand of \eqref{eq:bzint} by
$f(H(\vect{k}),\omega)$, and points in the PTR grid by $\vect{g}$. The
algorithm is simultaneously described for the unsymmetrized PTR for full BZ calculations,
and the symmetrized PTR described in Appendix \ref{app:ibz}, with the
option of using a symmetrized grid denoted by writing ``symmetrized'' in
parentheses. We denote the
weights in either PTR grid by $w_{\vect{g}}$. For the unsymmetrized $N^d$-point grid
on the full BZ $[-\pi,\pi]^d$, we have $w_{\vect{g}} =
(2\pi/N)^d$, and for the symmetrized PTR grid, the weights can be computed using
Algorithm \ref{alg:ptrsym}.

\begin{algorithm}[H]
  \caption{Automatic PTR to compute $G(\omega_j)$ at a collection of frequencies
  $\omega_j$}
  \label{alg:ptralt}
  \textbf{Inputs:} $\eta > 0$, $\varepsilon > 0$, positive integers
  $N$ and $\Delta N$, $n_\omega$ points $\omega_j$, $H_{\vect{R}}$ \\
  \textbf{Output:} $G(\omega_j)$, $j = 1,\ldots,n_\omega$, correct to $\varepsilon$ accuracy
  \begin{algorithmic}
    \State Initialize $J = \{1,\ldots,n_\omega\}$
    \State Initialize $\xtrial{G}(\omega_j) = 0$ for all
    $j \in J$
    \State Compute $w_{\vect{g}}$ for (symmetrized) $N^d$-point PTR grid
    \Loop{ over $\vect{g}$ in (symmetrized) $N^d$-point PTR grid}
      \State Evaluate $H(\vect{g})$
      \For{$j=1,\ldots,n_\omega$}
        \State $\xtrial{G}(\omega_j) \mathrel{{+}{=}} w_{\vect{g}}
        f(H(\vect{g}),\omega_j)$
      \EndFor
    \EndLoop
  \While{$J \neq \emptyset$}
    \State $N \gets N + \Delta N$
    \State Compute $w_{\vect{g}}$ for (symmetrized) $N^d$-point PTR grid
    \State Initialize $\xtest{G}(\omega_j) = 0$ for all $j \in J$
    \Loop{ over $\vect{g}$ in (symmetrized) $N^d$-point PTR grid}
      \State Evaluate $H(\vect{g})$
      \For{$j \in J$}
        \State $\xtest{G}(\omega_j) \mathrel{{+}{=}} w_{\vect{g}}
        f(H(\vect{g}),\omega_j)$
      \EndFor
    \EndLoop
    \For{$j \in J$}
      \State $e \gets \abs{\xtrial{G}(\omega_j)-\xtest{G}(\omega_j)}$
      \If{$e < \varepsilon$}
        \State Set $G(\omega_j) = \xtest{G}(\omega_j)$
        \State $J \gets J \, \backslash \,\{j\}$
      \Else
        \State $\xtrial{G}(\omega_j) \gets \xtest{G}(\omega_j)$
      \EndIf
    \EndFor
  \EndWhile
  \end{algorithmic}
\end{algorithm}

Algorithm \ref{alg:ptralt} cannot be used if the evaluation frequencies
$\omega_j$ are not specified in advance, as is the case when the adaptive
frequency interpolation scheme described in Section \ref{sec:adapinterp} is
used. Since use of an adaptive frequency grid gives a reduction in the number of
points required to resolve $G(\omega)$ from $\OO{\eta^{-1}}$ to $\OO{\log
(\eta^{-1})}$ compared with a uniform grid, the advantages provided by Algorithm
\ref{alg:ptralt} are typically insufficient to justify its use if it is necessary to fully
resolve $G(\omega)$, rather than evaluate it on a pre-specified grid.

\section{Linear tetrahedron method with \texorpdfstring{$\eta>0$}{positive
eta}} \label{app:tm}

The LTM is a popular BZ
integration method in electronic structure calculations, and we therefore
discuss it and compare it with our approach in this Appendix. We first
emphasize that the classical use case of the LTM is for BZ integrals
with integrands involving large Hamiltonian matrices which cannot be
efficiently evaluated on the fly using Wannier interpolation, and with
zero broadening $\eta = 0^+$. This is a
challenging setting, in which the LTM has significant
strengths, but is not the focus of this article. In the setting of 
a Wannier-interpolated Hamiltonian $H(\vect{k})$ that can be evaluated
efficiently on the fly, and a non-zero but
possibly small broadening $\eta > 0$ (or, more broadly, a self-energy),  more research is necessary
to understand how the LTM should be implemented, whether it is robust, and
what its asymptotic cost and performance are as $\eta \to 0^+$.

The idea of the LTM for integrals of the form \eqref{eq:bzint} can be
summarized as follows: the BZ is partitioned into tetrahedra of diameter
approximately $\Delta k$,
$H(\vect{k})$ is replaced on each tetrahedron by a linear
interpolant through its values at the vertices, and the resulting
integrals are computed analytically. If $H(\vect{k})$ is scalar-valued,
then the required analytical formula can be derived, even for $\eta =
0^+$, and if it is matrix-valued, then the scalar formula can be applied
to its diagonalization.

The LTM has the low-order accuracy $\OO{(\Delta
k)^2}$, a significant limitation to its use as a black-box method
to efficiently deliver a user-specified accuracy.
More importantly, in the case of
\eqref{eq:gfun} with a non-zero self-energy, diagonalizing
$H(\vect{k}) + \Sigma(\omega)$ yields complex eigenvalues, and it is
unclear how the LTM should be implemented in that case.
Ref.~\cite{haule10} describes the issue in terms of an ambiguity
in the ordering of complex eigenvalues in the vicinity of
band crossings, and proposes a heuristic algorithm to address the
problem, but the convergence
properties of the resulting scheme are not known.

Even if we focus on the case of \eqref{eq:bzint}
with scalar-valued $H(\vect{k})$, questions of robustness and efficiency remain. 
To illustrate one such concern, we show that in the classical case $\eta = 0^+$,
the LTM can yield a density of states with arbitrarily
large error at some $\omega$, no matter how small $\Delta k$ is, i.e.,
the LTM
does not converge uniformly. We will then return to the present case of
interest, $\eta > 0$. 

Consider the calculation of the density of states
\[A_0(\omega) = \intbz d^3 \vect{k} \, \delta\paren{\omega-H(\vect{k})},\]
the $\eta = 0^+$ limit of \eqref{eq:dos},
for the simple case $H(\vect{k}) = \cos(k_x) + \cos(k_y) + \cos(k_z)$.
$A_0(\omega)$ has a Van Hove singularity at $\omega_0 = 3$ due to
the extremum at $\vect{k} = \vect{0}$. It is straightforward to show
that $A_0(\omega) \sim \sqrt{\omega_0 - \omega}$ as $\omega \to
\omega_0^-$. For a fixed partition of $\bz$ by tetrahedra of
diameter $\Delta k$, consider $\omega$ sufficiently close to $\omega_0$
that
\[A_0(\omega) = \int_T d^3 \vect{k} \, \delta\paren{\omega-H(\vect{k})}\]
for a single tetrahedron $T$, i.e., no other tetrahedron contributes to
$A_0(\omega)$. Within $T$, the LTM makes the approximation $H(\vect{k}) \approx \vect{\nu} \cdot \vect{k} +
\alpha$. This yields 
\[A_0(\omega) \approx \int_T d^3 \vect{k} \,
\delta\paren{\omega-\paren{\vect{\nu} \cdot \vect{k} + \alpha}} =
\frac{s(\omega)}{\abs{\vect{\nu}}},\]
with
\[s(\omega) = \int_{\{\vect{k} \in T \, | \, \vect{\nu}\cdot\vect{k} + \alpha = \omega\}} dS(\vect{k})\]
the surface area of the plane
$\vect{\nu}\cdot\vect{k} + \alpha = \omega$ in $T$. If $\omega$ is
chosen so that $s(\omega) > 0$, then the resulting
approximation of $A_0(\omega)$, which should be $\sim
\sqrt{\omega_0-\omega}$, can be arbitrarily large, depending on how
small $\abs{\vect{\nu}}$ is. $\abs{\vect{\nu}}$ is a measure of
the flatness of the linear interpolant of $H(\vect{k})$ through the
vertices of $T$, and
therefore depends sensitively on the specific location of $T$
relative to $\vect{k}=\vect{0}$.
By simple translation of the mesh, one can arrange $\abs{\vect{\nu}}$ to be
arbitrarily small, leading to an arbitrarily large spurious peak in
$A_0(\omega)$. This phenomenon may be encountered in the vicinity
of any band extremum, and the tetrahedral mesh cannot in general
be arranged to avoid it without significant customization of
the mesh for each given band structure. Although the range in $\omega$
over which such a problem is encountered decreases as $\Delta k$
decreases, the expected magnitude of the error increases, since
$H(\vect{k})$ becomes flatter in the tetrahedron containing the band
edge. The problem persists no matter how small $\Delta k$ is taken.

For the $\eta = 0^+$ case, the LTM cannot therefore be used to build a
robust and automatic BZ integration algorithm providing results to a
user-specified accuracy without correcting failure modes of this nature.
When $\eta > 0$, the analysis is less straightforward. To the authors'
knowledge, it has not been determined in the literature whether or not
the problem is eliminated in this setting, and more generally whether or
not the LTM is convergent.  Even assuming it is convergent, the
computational complexity of the method as $\eta \to 0^+$ is not obvious.
This question merits further investigation, but our preliminary
experiments indicate that the computational effort required to achieve a
given accuracy is not uniform in $\eta$, and varies significantly with
the choice of $\omega$.

\bibliographystyle{ieeetr}
\bibliography{autobz}
\end{document}